\documentclass[runningheads]{llncs}
\usepackage{amsmath}
\usepackage{amssymb}
\usepackage{booktabs}
\usepackage{graphicx}
\usepackage[hidelinks]{hyperref}
\usepackage[capitalise,noabbrev]{cleveref}

\title{Repeated-Game Security for Restaking-Based Verifiable Inference}
\titlerunning{Repeated-Game Security for Restaking-Based Verifiable Inference}

\author{Zhenhang Shang \and Yingzhe Yu \and Kani Chen}
\authorrunning{Z. Shang et al.}
\institute{The Hong Kong University of Science and Technology, Hong Kong, China\\
\email{zshangab@connect.ust.hk}}

\newcommand{\Stake}{S}                 % stake amount
\newcommand{\Chal}{c}                  % challenge probability
\newcommand{\Slash}{\alpha}            % slashing fraction (alpha avoids collision with per-query reward r)
\newcommand{\Disc}{\delta}             % discount factor
\newcommand{\Cheat}{p}                 % cheating probability
\newcommand{\Rep}{\rho}                % reputation score
\newcommand{\Saving}{\Delta}           % per-request cost saving from cheating
\newcommand{\Horizon}{T}               % game length
\newcommand{\Vest}{\tau}               % vesting period
\newcommand{\Eps}{\varepsilon}

\newcommand{\HONEST}{\mathsf{H}}
\newcommand{\CHEAT}{\mathsf{C}}

\newcommand{\E}{\mathbb{E}}
\newcommand{\Prob}{\mathbb{P}}

\begin{document}

\maketitle

\begin{abstract}
Restaking-based protocols are emerging as a practical way to provide verifiable LLM inference without the high proving cost of zkML or the hardware trust assumptions of TEEs. Their security is usually justified by a one-round slashing condition: a rational provider should not cheat when the expected penalty exceeds the per-query cost saving. This paper shows that this condition can overstate security once inference is supplied repeatedly under the same stake. We model restaking-based verifiable inference as a discounted repeated game and identify a repeated-game gap caused by proportional slashing: detected deviations reduce the stake exposed to future penalties, while the saving from cheaper inference is earned again on every query. We derive the gap in closed form, show that it is robust to minimum-stake ejection, and extend it to memoryless bounded-slashing protocols, a class that captures deployed verifiable-inference systems. We then propose a deployable mechanism combining history-dependent challenges, reputation-weighted slashing, and stake vesting that restores infinite-horizon subgame-perfect incentive compatibility against stationary mixed-strategy deviations above an explicit discount-factor threshold, without per-query cryptographic verification. Two further results support the construction: the first direct measurement of the audit signal's detectability response on open-weight models from 0.5B to 14B parameters, which is concave as the mechanism requires in all nine (contracted, substitute) pairs, with amplitude increasing in the cost ratio between the two models and measured security thresholds well below our conservative calibration; and a Stackelberg audit-budget analysis showing that signal responsiveness substitutes for baseline auditing, cutting the required baseline audit rate by $2.6\times$ at target discount factor $0.95$. Calibrating to deployed parameters, surveyed protocols pass the one-round IC test but fail repeated-game IC for $\Disc \in [0.92,0.98]$, with deviation-profit fractions of $1.5\%$--$8\%$; in simulation, the mechanism cuts deviation profit by $31\%$--$54\%$ while preserving low-latency economic verification.
\end{abstract}

% #############################################################################
% #############################################################################
%                                MAIN BODY
% #############################################################################
% #############################################################################

\section{Introduction}
\label{sec:introduction}

% \paragraph*{Motivation.}
Restaking has become a central economic layer for verifiable services on Ethereum: EigenLayer supports a growing set of AVSs securing billions of dollars in restaked capital~\cite{eigenlayer-whitepaper}, and verifiable LLM inference is among the most visible workloads. Protocols such as EigenAI~\cite{eigenai2026}, VeriLLM~\cite{verillm2025}, and Sertn AVS~\cite{sertn-avs} follow a common template in which providers post stake, serve low-latency inference, and face randomized re-execution challenges, with disagreement punished through slashing. The appeal is economic rather than purely cryptographic: zkML systems~\cite{ezkl,verifiable-evaluations-zksnark-2024} remain orders of magnitude more expensive than native execution at frontier scale, while TEE-based approaches~\cite{apple-pcc} shift the trust assumption to hardware. The design's viability depends on whether the security calculation is calibrated to the way inference is actually supplied.

The standard analysis largely inherits the logic of proof-of-stake slashing: honest behavior is treated as optimal whenever the expected penalty $\Chal \Slash \Stake$ exceeds the cheating gain $\Saving$, as formalized in \cref{def:1ric}. This condition is natural for a single isolated decision, but verifiable inference is supplied repeatedly under the same economic contract---a production provider may answer $10^4$ to $10^6$ queries before exiting---so the relevant object is a discounted repeated game rather than a one-shot incentive constraint~\cite{mailath-samuelson-2006,fudenberg-maskin-1986}. Stage-game incentive compatibility is then only a local condition: it does not control deviations that exploit the continuation value of the relationship, including randomized cheating that remains profitable even under one-round slack. This repeated-interaction gap is absent from the protocol-level analyses of verifiable inference systems that we examined.

We address this gap by modeling verifiable inference as a discounted repeated game between a rational provider and an on-chain protocol. The analysis identifies when the single-round condition ceases to be sufficient, measures the resulting shortfall in economic security, and leads to a deployable mechanism that restores subgame-perfect incentives at realistic operator patience, as follows.
\begin{itemize}
\item We introduce a repeated-game security model for restaking-based verifiable inference, defining $\Horizon$-round subgame-perfect incentive compatibility and an $\Eps$-approximate variant.
\item We show that the standard one-round slashing condition can overstate security: with proportional slashing $\Slash < 1$, a protocol can satisfy one-round IC with slack and still fail $\infty$-SPIC whenever $\eta < \frac{\Disc \Chal \Slash \Saving}{1-\Disc}$, with deviation profit in closed form. The gap survives minimum-stake ejection and extends to memoryless bounded-slashing protocols, a class covering the deployed systems in our survey.
\item We construct a deployable mechanism combining history-dependent challenges, reputation-weighted slashing, and stake vesting that restores repeated-game incentives against stationary mixed-strategy deviations whenever $\Disc \geq \Disc^\ast$, the unique solution to \eqref{eq:disc-star}, without per-query cryptographic verification and within the EigenLayer ServiceManager interface.
\item We provide the first direct measurement of the detectability response $\mu(\Cheat)$ the mechanism relies on: on open-weight Qwen2.5 models from 0.5B to 14B parameters across nine (contracted, substitute) pairs, the audit signal is concave in the cheating rate as required ($q\in[0.70,0.89]$), its amplitude grows with the contracted-to-substitute cost ratio as $\min(1,0.45x^{0.41})$, and the measured responses lower $\Disc^\ast$ to $0.79$--$0.92$, below the conservative calibration (\cref{app:calibration}).
\item We calibrate the model to deployed parameters: for $\Disc \in [0.92,0.98]$, every surveyed protocol passes the one-round IC test but fails repeated-game IC, with deviation-profit fractions of $1.5\%$--$8\%$. The mechanism's required stake stays bounded as $\Disc \to 1$ (\cref{cor:asymptotic}), while the strengthened one-round bound diverges. A Stackelberg audit-budget analysis shows that signal responsiveness substitutes for baseline auditing, cutting the required baseline audit rate by $2.6\times$ at $\Disc_{\mathrm{target}}=0.95$ (\cref{sec:discussion}).
\end{itemize}

\section{Background and Related Work}
\label{sec:background}
Restaking turns staked ETH into reusable economic security for external services: in EigenLayer, validators or LST holders pledge stake to Actively Validated Services, with slashing enforced through the ServiceManager interface~\cite{eigenlayer-whitepaper}. The standard security argument follows proof-of-stake slashing~\cite{ethereum-slashing} as a single-round condition: deviation is unattractive when expected slashing exceeds the gain. Buterin's critique of restaking targets systemic risks such as capital reuse and correlated slashing~\cite{buterin-restaking-risk}; our focus is orthogonal: even with fixed capital and slashing rules, repeated interaction can invalidate a one-round-correct calculation. Verifiable LLM inference is a natural setting for this distinction: protocols such as EigenAI~\cite{eigenai2026}, VeriLLM~\cite{verillm2025}, and Sertn AVS~\cite{sertn-avs} all rely on randomized re-execution, provider stake, and slashing. Our class result shows that the gap is a property of this mechanism family, not of a particular implementation.

The analysis draws on repeated games with imperfect public monitoring~\cite{fudenberg-maskin-1986,abreu-pearce-stacchetti-1990,mailath-samuelson-2006} and inspection games~\cite{avenhaus-canty-1996}. Daian et al.~\cite{daian-2020-flashboys} show that single-round incentive reasoning can misstate MEV incentives; single-round audit models also appear in Game of Coding~\cite{game-of-coding} and in trustless verification under uncertainty~\cite{trustless-uncertainty-2025}. Our setting differs because the same stake supports many future queries: proportional slashing lowers the penalty base after a detected deviation, while the saving is earned again each round---an intertemporal mismatch invisible in one-round models.

The closest repeated-game analysis in the blockchain literature is Arbitrum's dispute game~\cite{kalodner-arbitrum-2018}, where one correct challenger deters false assertions because a caught cheater loses its bond in full. Two features separate that setting from ours: Arbitrum's deterrent relies on full confiscation, the full-slashing regime of \cref{rem:full-slashing}, whereas restaking-based inference uses partial slashing; and rollup disputes resolve within a fixed window, whereas an inference provider answers $10^4$ to $10^6$ queries under a slowly depreciating stake.

\section{Model and Definitions}
\label{sec:model}

\subsection{System Architecture}
\label{sec:model-arch}

We model a restaking-based verifiable LLM inference service with four actors, following EigenAI, VeriLLM, and Sertn AVS. A requester submits a query and pays a fee. A provider $P$ posts stake $\Stake$, serves inference for a per-query reward $r$, and privately chooses $a \in \{\HONEST,\CHEAT\}$: under $\HONEST$ it runs the contracted model, under $\CHEAT$ a cheaper substitute that saves computation cost $\Saving$. A verifier $V$ re-executes the query with probability $\Chal$ and reports the outcome to the on-chain mechanism $M$, which schedules audits, records outcomes, and slashes a fraction $\Slash \in (0,1]$ of the provider's current stake upon detected cheating; slashing is irreversible. The baseline assumes an honest verifier and perfect audit correctness; \cref{sec:eq-robust} discusses collusion and noisy audits.

\subsection{The Stage Game}
\label{sec:model-stage}

In the one-shot stage game $\mathcal{G}$, the provider chooses $a\in A_P=\{\HONEST,\CHEAT\}$. The binary action space is without loss for the impossibility result, since richer effort choices reduce to the most profitable detectable deviation (\cref{sec:gap-class}). Let $\xi \sim \mathrm{Bernoulli}(\Chal)$ be the audit indicator, sampled independently of the provider's action. The provider's payoff is $u_P(a,\Stake,\xi)=r$ if $a=\HONEST$, and $u_P(a,\Stake,\xi)=r+\Saving-\xi \Slash \Stake$ if $a=\CHEAT$, so in expectation honest execution earns $r$, while cheating earns $r+\Saving-\Chal\Slash\Stake$. The reward is paid when the response is delivered, before the randomized audit settles: deferring payment until after a possible re-execution would reintroduce the latency this design exists to avoid, so slashing is the only ex post instrument and a detected deviation still nets the reward $r$.

The provider's action is private; the protocol observes only the audit indicator and its outcome: an audited query reveals cheating perfectly, while an unaudited query carries no information. We write the public history before round $t$ as $h_t^{\mathrm{pub}}=\bigl((\xi_\tau,\mathrm{outcome}_\tau)\bigr)_{\tau<t}$. The corresponding one-round security condition is the standard cryptoeconomic slashing constraint.

\begin{definition}[1-Round Incentive Compatibility]
\label{def:1ric}
Parameters $(\Stake,\Chal,\Slash,\Saving)$ satisfy \emph{1-Round Incentive Compatibility} (\emph{1-Round IC}) if honest execution weakly dominates cheating in the stage game, equivalently $\Chal \Slash \Stake \geq \Saving$.
\end{definition}

This is the incentive test used, often implicitly, in restaking-based verifiable inference. \cref{sec:gap} shows that it does not imply security once the same provider interacts with the protocol repeatedly.

\subsection{The Repeated Game}
\label{sec:model-repeated}

The repeated game $\mathcal{G}^{\Horizon}$, with $\Horizon \in \mathbb{N} \cup \{\infty\}$, consists of $\Horizon$ rounds of the stage game in \cref{sec:model-stage}. The stake state is carried across rounds. In round $t$, the provider chooses an action $a_t \in \{\HONEST,\CHEAT\}$ and the protocol samples an audit indicator $\text{audit}_t \in \{0,1\}$. For a provider strategy $\sigma$, the discounted payoff is $U_P(\sigma)=\E_\sigma\bigl[\sum_{t=0}^{\Horizon-1}\Disc^t u_P(a_t,\Stake_t,\text{audit}_t)\bigr]$.

A mixed stationary strategy cheats with a fixed probability $\Cheat \in [0,1]$ every round; a public behavioral strategy maps public histories to cheating probabilities, $\sigma_t:\mathcal{H}^{\mathrm{pub}}_t\to[0,1]$; general behavioral strategies may also use the provider's private history. By the one-shot deviation principle (\cref{rem:osdp}), it suffices to analyze public behavioral deviations on the equilibrium path.

The discount factor $\Disc \in (0,1)$ summarizes the provider's continuation value. One round corresponds to one settlement interval of the service, which we take to be one day: audits are sampled per query, but slashing and reward settlement are batched. We decompose the factor into an operator-retention component and a capital-cost component,
\begin{equation}
\Disc
=
(1-r_{\mathrm{out}})(1-r_{\mathrm{cap}}),
\label{eq:disc-decomp}
\end{equation}
where $r_{\mathrm{out}}$ is the per-round opt-out hazard and $r_{\mathrm{cap}}$ the per-round opportunity cost of locked capital. Two features matter for interpretation: $\Disc$ is not subjective time preference but the probability that the staking relationship survives another round---pure time preference at a $5\%$--$10\%$ annual rate would imply a daily factor above $0.9997$, negligible at this frequency---and the binding component is the exit hazard, estimated at $r_{\mathrm{out}}\in[0.5\%,1.5\%]$ per day from EigenLayer deregistration events, with $r_{\mathrm{cap}}\in[0.013\%,0.020\%]$ per day from on-chain low-risk yields, giving $\Disc \in [0.92,0.98]$ at one-day rounds. The corresponding expected operating horizon of $1/(1-\Disc)\approx 12$ to $50$ settlement intervals is conservative relative to observed operator tenures of several months (\cref{app:calibration}).

\subsection{Equilibrium Concepts}
\label{sec:model-equilibrium-concepts}

We use standard concepts for repeated games with imperfect public monitoring~\cite{mailath-samuelson-2006}. A \emph{public strategy} maps each public history to a (possibly mixed) action, $\sigma_t:\mathcal{H}^{\mathrm{pub}}_t\to[0,1]$.

\begin{definition}[SPE and PPE]
\label{def:spe}
\label{def:ppe}
A strategy profile $\sigma$ is a \emph{subgame perfect equilibrium} (SPE) of $\mathcal{G}^{\Horizon}$ if, after every history $h_t$, the continuation strategies $\sigma\!\restriction_{h_t}$ form a Nash equilibrium of the continuation game. It is a \emph{public-perfect equilibrium} (PPE)~\cite{abreu-pearce-stacchetti-1990} if all strategies are public and, after every public history $h_t^{\mathrm{pub}}$, the continuation strategies form a Nash equilibrium of the continuation game.
\end{definition}

\begin{remark}[One-shot deviation principle]
\label{rem:osdp}
Under discounting, a public strategy profile is a PPE if and only if no player has a profitable \emph{one-shot deviation}: a deviation at a single public history, followed by a return to the prescribed continuation strategy~\cite[Thm.~2.1.1]{mailath-samuelson-2006}. All equilibrium verification in this paper proceeds by ruling out profitable one-shot deviations.
\end{remark}

Because the provider's action is private and only audit outcomes are public, PPE is the natural refinement; the mechanism state in our construction is itself public, so public strategies are without loss for the analysis of \cref{sec:equilibrium}. The folk theorem explains why repeated incentives need not follow from one-round incentives~\cite{abreu-pearce-stacchetti-1990,fudenberg-maskin-1986}, but gives no deployable protocol for a given finite patience level; our construction characterizes the required threshold $\Disc^\ast$ in closed form.

\subsection{Incentive Compatibility Notions}
\label{sec:model-security-defs}

The stage game has a single strategic player, so the equilibrium concepts above collapse to optimality conditions on the provider: the security notions below are properties of a \emph{mechanism design} problem in which the protocol commits to an audit and slashing rule and the provider best-responds (\cref{sec:discussion} discusses a strategic audit policy). Subgame perfection then requires that honesty remains optimal after every history.

\begin{definition}[$\Horizon$-SPIC and $\Eps$-Approximate $\Horizon$-SPIC]
\label{def:tspic}
\label{def:eps-tspic}
Parameters $(\Stake_0,\Chal,\Slash,\Saving)$ satisfy \emph{$\Horizon$-Round Subgame-Perfect Incentive Compatibility}, or $\Horizon$-SPIC, under discount factor $\Disc$ if the always-honest strategy is a subgame perfect equilibrium of the repeated game $\mathcal{G}^{\Horizon}$. They satisfy \emph{$\Eps$-Approximate $\Horizon$-SPIC} if no strategy improves on always-honest behavior by more than an $\Eps$ fraction of the honest payoff,
$
\sup_{\sigma \neq \HONEST^{\Horizon}}
\frac{
U_P(\sigma)-U_P(\HONEST^{\Horizon})
}{
U_P(\HONEST^{\Horizon})
}
\leq
\Eps$.
\end{definition}
1-Round IC is necessary but not sufficient for $\Horizon$-SPIC (\cref{prop:hierarchy} in \cref{app:proofs}); the gap between them is the subject of the next section.

\section{The Repeated-Game Gap}
\label{sec:gap}
\subsection{A Mixed-Strategy Deviation}
\label{sec:gap-deviation}

Consider a stationary mixed strategy $\sigma\equiv\Cheat\in[0,1]$, cheating independently each round with probability $\Cheat$. If stake were fixed at $\Stake$, the expected per-round payoff would be $\E[u_P \mid \Cheat]=r+\Cheat(\Saving-\Chal\Slash\Stake)$; under 1-Round IC, $\Chal\Slash\Stake\geq\Saving$, so this one-round calculation makes mixed cheating appear unprofitable.

The calculation fails because proportional slashing changes the future stake at risk: for $\Slash<1$, the relevant regime in practice,\footnote{Ethereum slashes only a fraction of validator stake, and EigenLayer AVS designs typically use partial slashing rates. \cref{rem:full-slashing} discusses the full-slashing limit.} stake evolves as
\begin{equation}
\Stake_{t+1}=\Stake_t\bigl(1-\Slash\mathbf{1}[a_t=\CHEAT,\text{ audited}]\bigr).
\label{eq:stake-dynamics}
\end{equation}
Under stationary cheating, $\E[\Stake_t\mid\Cheat]=\Stake_0(1-\Cheat\Chal\Slash)^t$: slashing exposure decays geometrically, while the cheating gain $\Saving$ is earned again each round, so a condition binding at the initial stake can fail along the continuation path, giving
\begin{equation}
\E[u_t\mid\Cheat]=r+\Cheat\bigl(\Saving-\Chal\Slash\Stake_0(1-\Cheat\Chal\Slash)^t\bigr).
\label{eq:stage-with-decay}
\end{equation}

\begin{lemma}[Mixed-Strategy Gap]
\label{lem:gap}
Fix $(\Stake_0,\Chal,\Slash,\Saving,\Disc)\in\mathbb{R}_{>0}^4\times(0,1)$ with $\Slash\in(0,1)$. Suppose 1-Round IC binds at the initial stake, $\Chal\Slash\Stake_0=\Saving$. Under \eqref{eq:stake-dynamics}, the always-cheat strategy $\Cheat=1$ yields a strictly higher discounted payoff than always-honest behavior:
\begin{equation}
\Delta U
=
U_P(\Cheat=1)-U_P(\HONEST)
=
\frac{\Disc\Chal\Slash\Saving}
{(1-\Disc)(1-\Disc+\Disc\Chal\Slash)}
>0 .
\label{eq:gap-formula}
\end{equation}
Relative to $U_P(\HONEST)=r/(1-\Disc)$,
\begin{equation}
\frac{\Delta U}{U_P(\HONEST)}
=
\frac{\Disc(\Chal\Slash)^2}
{1-\Disc+\Disc\Chal\Slash}
\cdot
\frac{\Stake_0}{r}.
\label{eq:gap-fraction}
\end{equation}
\end{lemma}

\begin{proof}
At $\Cheat=1$, an audited cheating round occurs with probability $\Chal$ each round, so the expected stake is $\E[\Stake_t]=\Stake_0(1-\Chal\Slash)^t$. Using the binding condition $\Chal\Slash\Stake_0=\Saving$, \eqref{eq:stage-with-decay} gives $\E[u_t\mid\Cheat=1]=r+\Saving\bigl(1-(1-\Chal\Slash)^t\bigr)$. Discounting and summing,
\[
U_P(\Cheat=1)
=
\frac{r}{1-\Disc}
+
\Saving
\left(
\frac{1}{1-\Disc}
-
\frac{1}{1-\Disc(1-\Chal\Slash)}
\right).
\]
Subtracting $U_P(\HONEST)=r/(1-\Disc)$ and simplifying the bracket yields \eqref{eq:gap-formula}. Both factors in the denominator are positive and the numerator is strictly positive for $\Slash>0$, so $\Delta U>0$. Substituting $\Saving=\Chal\Slash\Stake_0$ and dividing by $U_P(\HONEST)$ gives \eqref{eq:gap-fraction}.
\end{proof}

Full slashing closes the gap (\cref{rem:full-slashing} in \cref{app:proofs}): under $\Slash=1$, detection terminates the relationship, and the one-round condition alone suffices. As $\Disc\to1$, the relative gain approaches a term proportional to $\Chal\Slash\Stake_0/r$: the gap is largest for patient providers with high stake relative to per-query rewards, the deployed AVS regime.

Minimum-stake ejection does not remove the gap: a deviator simply cheats until ejection, and since ejection requires $N^\ast$ audited failures, the provider retains nearly all of the gain in \cref{lem:gap}. \cref{thm:gap-eject} in \cref{app:proofs} formalizes this: the cheat-until-ejection payoff loses only a factor $\zeta=\E[\Disc^\tau]$ relative to \eqref{eq:gap-fraction}, where $\tau$ is the negative-binomial ejection time. For deployed ranges, $\zeta \in [10^{-60},10^{-9}]$ (\cref{rem:eject-bite}), negligible relative to the $10^{-2}$ scale of \eqref{eq:gap-fraction}: ejection affects only the terminal boundary condition.

\subsection{The Main Gap Theorem}
\label{sec:gap-theorem}

We now allow the one-round condition to hold with strict slack and characterize the extra slashing budget needed to rule out the repeated-game deviation.

\begin{theorem}[Repeated-Game Gap with Slack]
\label{thm:gap}
Let $(\Stake_0,\Chal,\Slash,\Saving,\Disc)$ be protocol parameters with $\Slash \in (0,1)$ and $\Disc \in (0,1)$. Define the 1-Round IC slack by
\begin{equation}
\eta
=
\Chal\Slash\Stake_0-\Saving
\geq
0 .
\label{eq:slack-def}
\end{equation}
If
\begin{equation}
\eta
<
\frac{\Disc\Chal\Slash\Saving}{1-\Disc},
\label{eq:slack-threshold}
\end{equation}
then the parameters do not satisfy $\infty$-SPIC under discount factor $\Disc$. Equivalently, $\infty$-SPIC requires
\begin{equation}
\Chal\Slash\Stake_0
\geq
\Saving
\left(
1+
\frac{\Disc\Chal\Slash}{1-\Disc}
\right).
\label{eq:tspic-condition}
\end{equation}
\end{theorem}

\begin{proof}
Consider the always-cheat deviation $\Cheat=1$. An audited cheating round occurs with probability $\Chal$ each round, so $\E[\Stake_t]=\Stake_0(1-\Chal\Slash)^t$, and the round-$t$ expected payoff is $r+\Saving-\Chal\Slash\Stake_0(1-\Chal\Slash)^t$. Summing discounted payoffs and subtracting $U_P(\HONEST)=r/(1-\Disc)$ gives
\begin{equation}
\Delta U(1)
=
\frac{\Saving}{1-\Disc}
-
\frac{\Chal\Slash\Stake_0}{1-\Disc(1-\Chal\Slash)}
=
\frac{\Saving}{1-\Disc}
-
\frac{\Saving+\eta}{1-\Disc+\Disc\Chal\Slash}
=
\frac{\Disc\Chal\Slash\Saving-\eta(1-\Disc)}
{(1-\Disc)(1-\Disc+\Disc\Chal\Slash)},
\end{equation}
where the second equality uses $\eta=\Chal\Slash\Stake_0-\Saving$ from \eqref{eq:slack-def}. The denominator is positive, so $\Delta U(1)>0$ exactly when \eqref{eq:slack-threshold} holds, and the parameters then fail $\infty$-SPIC. Rearranging the complementary condition $\eta\geq\Disc\Chal\Slash\Saving/(1-\Disc)$ gives \eqref{eq:tspic-condition}.
\end{proof}

Condition \eqref{eq:tspic-condition} thus requires more than the one-round slashing budget: the multiplier $1+\frac{\Disc\Chal\Slash}{1-\Disc}$ is approximately $1.038$ at $\Disc=0.95$, $\Chal=0.01$, $\Slash=0.2$, and diverges as $\Disc \to 1$, so proportional slashing alone cannot provide uniform repeated-game security in the patient-operator limit, motivating the mechanism in \cref{sec:mechanism}.

For a finite horizon $\Horizon$, the gap condition in \eqref{eq:slack-threshold} becomes $\eta<\frac{\Disc\Chal\Slash\Saving}{1-\Disc}\bigl(1-O(\Disc^\Horizon)\bigr)$, so the infinite-horizon conclusion is already a good approximation once $\Horizon \gtrsim 1/\log(1/\Disc)$, far below operational horizons of $10^4$ to $10^6$ queries.

\subsection{Quantification on Deployed Protocols}
\label{sec:gap-eigenai}

\cref{tab:eigenai-params} in \cref{app:eval-extra} instantiates the gap of \cref{lem:gap} on EigenAI~\cite{eigenai2026}, VeriLLM~\cite{verillm2025}, and Sertn AVS~\cite{sertn-avs}, which share the memoryless bounded-slashing structure of \cref{def:memoryless,thm:class,thm:gap-eject}; \cref{fig:gap-heatmap} maps the gap across the full $(\Chal\Slash,\Disc)$ plane. Computing $\Saving$ as the API-price differential between the contracted model and a quantized or distilled substitute, all three protocols admit profitable repeated-game deviations at the percent scale---$5.4\%$--$8.2\%$ for EigenAI, $2.6\%$--$3.9\%$ for VeriLLM, and $1.5\%$--$2.5\%$ for Sertn AVS at $\Disc=0.95$---despite satisfying the one-round IC test; the simulations of \cref{sec:evaluation} show that adversary A4 realizes these gains under the calibrated dynamics.

\subsection{A Generic Class Result}
\label{sec:gap-class}

\cref{thm:gap} is stated for the stage game in \cref{sec:model-stage}; we next extend the argument to a broader class of slashing protocols that captures the common structure of deployed verifiable-inference systems.

\begin{definition}[Memoryless Bounded-Slashing Protocol]
\label{def:memoryless}
A slashing protocol $\Pi$ is \emph{$(\Chal,\Slash_{\max})$-memoryless-bounded} if, in every round $t$ and after every history $h_t$: (i)~the audit indicator is distributed as $\mathrm{Bernoulli}(\Chal)$ and is independent of $h_t$; (ii)~any detected cheating event triggers slashing of at most $\Slash_{\max}\Stake_t$, where $\Slash_{\max} \in (0,1)$; and (iii)~challenge probabilities, slashing magnitudes, and other protocol decisions depend on the history only through the current stake $\Stake_t$.
\end{definition}

\begin{theorem}[Generic Gap for Memoryless Protocols]
\label{thm:class}
Let $\Pi$ be a $(\Chal,\Slash_{\max})$-memoryless-bounded protocol with $\Slash_{\max} \in (0,1)$. For operational parameters $(\Stake_0,\Saving,\Disc)$ with $\Saving>0$, $\Disc \in (0,1)$, and $\Stake_0>0$, define the one-round slack by
$
\eta
=
\Chal\Slash_{\max}\Stake_0-\Saving .
$
If
$
\eta
<
\frac{\Disc\Chal\Slash_{\max}\Saving}{1-\Disc},
$
then $\Pi$ fails $\infty$-SPIC under discount factor $\Disc$. The deviation-profit fraction is bounded below by
\[
\Omega
\left(
\frac{
\Disc(\Chal\Slash_{\max})^2
}{
1-\Disc+\Disc\Chal\Slash_{\max}
}
\cdot
\frac{\Stake_0}{r}
\right).
\]
\end{theorem}

\begin{proof}
The most favorable case for the protocol is one in which every detected cheating event uses the full allowable penalty $\Slash_{\max}\Stake_t$. Any smaller penalty only increases the provider's deviation payoff. Under this best-case penalty rule, the protocol reduces to the stage game in \cref{sec:model-stage} with $\Slash=\Slash_{\max}$, so \cref{thm:gap} applies and yields both the threshold and the profit-fraction bound.
\end{proof}

The definition covers the slashing logic of EigenAI, VeriLLM, and Sertn AVS under the ranges in \cref{tab:eigenai-params}: the failure follows from memoryless audits and bounded proportional slashing, not from a particular implementation. The mechanism in \cref{sec:mechanism} avoids the result by making audit intensity and penalties depend on the provider's public history.

The binary action model also extends to continuous cheating intensities with concave windfall and detection probability (\cref{rem:continuous} in \cref{app:proofs}): concentrating the same effective deviation at full intensity preserves the gap bounds, so continuous effort choices change which deviation is worst, not whether the gap exists.

\section{Mechanism Design for Long-Run IC}
\label{sec:mechanism}

\subsection{Design Goals}
\label{sec:mech-goals}

\cref{thm:class} rules out long-run IC for memoryless bounded-slashing mechanisms at deployed patience levels; restoring incentives requires future treatment to depend on public history, without giving up low-latency economics. We design for four properties: an explicit $\Horizon$-SPIC threshold, no per-query cryptographic verification, implementation through the EigenLayer ServiceManager interface, and bounded exposure for new operators. The mechanism combines three instruments, each acting on the continuation payoff---the channel missing from \cref{thm:class}: history-dependent challenges raise audit intensity after suspicious behavior, reputation-weighted slashing keeps penalties from falling with the remaining stake, and vesting limits the value of cheating before exit.

\subsection{History-Dependent Challenge Schedule}
\label{sec:mech-challenge}

The protocol maintains a public suspicion score $\mathcal{S}_t\in\mathbb{R}_{\geq0}$ for each operator, updated from a normalized risk signal $z_t\in[0,1]$ combining (i) output-side discrepancy from a reference token-distribution sketch committed at registration, (ii) workload shifts in latency, prompt length, or response length relative to the operator's baseline, and (iii) cross-provider agreement on redundant queries. The AVS challenger committee normalizes these components; per-signal calibration appears in \cref{sec:eq-robust}.

The score follows an exponentially weighted update,
\begin{equation}
\mathcal{S}_{t+1}
=
(1-\lambda)\mathcal{S}_t+\lambda z_t,
\qquad
\mathcal{S}_0=0,
\qquad
\lambda\in(0,1].
\label{eq:susp-update}
\end{equation}

Audit probability increases with suspicion,
\begin{equation}
\Chal(\mathcal{S}_t)
=
\min\{1,\Chal_0+\beta\mathcal{S}_t\},
\qquad
\beta>0 .
\label{eq:chal-rule}
\end{equation}
The floor $\Chal_0>0$ preserves baseline auditability at zero suspicion.

For a stationary cheating rate $\Cheat$, let $\E[z_t\mid\Cheat]=\mu(\Cheat)$, with $\mu(0)=0$, $\mu'>0$, and $\mu$ concave on $[0,1]$. The steady-state audit rate is $\bar{\Chal}(\Cheat)=\Chal_0+\beta\mu(\Cheat)$, and the corresponding per-round payoff at stake $\Stake_t$ is $\E[u_P\mid\Cheat]=r+\Cheat\bigl(\Saving-\bar{\Chal}(\Cheat)\Slash\Stake_t\bigr)$. Because audit intensity rises with suspicious behavior, the expected penalty is no longer linear in the cheating rate: repeated cheating now raises future audit exposure, reducing the appeal of the corner deviation in \cref{lem:gap}.

\paragraph*{Key assumption: signal response.}
Deterrence depends on $\mu'(0)>0$, with $\beta\mu'(0)$ as the relevant deployment quantity: a substitute matching the contracted model across all signal channels while preserving the cost saving $\Saving$ would degrade the mechanism to the memoryless baseline. Multiple signals mitigate this, as detailed in \cref{sec:eq-robust}.

\subsection{Bayesian Reputation-Weighted Slashing}
\label{sec:mech-reputation}

The second component prevents slashing exposure from falling mechanically with stake, using a two-type reputation model~\cite{kreps-wilson-1982,milgrom-roberts-1982}: an operator is either a committed-honest type, which always plays $\HONEST$, or a mixed-strategic type, which cheats in each round with an unknown rate $\Cheat^\dagger \in (0,1]$. Let $\pi_0=\Prob[\theta=\CHEAT]$ denote the registration prior, and let $\Rep_t=\Prob[\theta=\CHEAT \mid h_t^{\mathrm{pub}}]$ be the public posterior after history $h_t^{\mathrm{pub}}$. Writing $s_t=1$ for a detected cheat and $s_t=0$ otherwise, with $\Prob[s_t=1 \mid \HONEST]=0$ and $\Prob[s_t=1 \mid \CHEAT]=\Cheat^\dagger$ conditional on an audit, Bayes' rule gives
\begin{equation}
\Rep_{t+1}
=
\begin{cases}
1, & \text{if } s_t=1, \\[4pt]
\dfrac{\Rep_t(1-\Cheat^\dagger)}
{\Rep_t(1-\Cheat^\dagger)+(1-\Rep_t)},
& \text{if } s_t=0 \text{ and the round was audited}, \\[12pt]
\Rep_t, & \text{if the round was not audited}.
\end{cases}
\label{eq:rep-update}
\end{equation}
A detected cheat is conclusive in the baseline model; a clean audit lowers the posterior, and an unaudited round carries no new information. The update in \eqref{eq:rep-update} depends on the unobserved $\Cheat^\dagger$, so we use a robust posterior in the sense of robust mechanism design~\cite{bergemann-morris-2005}: at deployment, the AVS fixes an interval $[\underline\Cheat,1]$, where $\underline\Cheat$ is the smallest cheating rate it commits to deter, and works with the worst-case posterior over that interval.
\begin{definition}[Robust Bayesian Posterior]
\label{def:robust-rep}
Given public audit history $h_t^{\mathrm{pub}}$ and deployment interval $[\underline\Cheat,1]$, the \emph{robust posterior} is $\bar\Rep_t=\sup_{\Cheat^\dagger \in [\underline\Cheat,1]}\Prob[\theta=\CHEAT\mid h_t^{\mathrm{pub}},\Cheat^\dagger]$.
\end{definition}

The robust posterior admits a closed form (\cref{lem:robust-update} in \cref{app:protocol-lemmas}): it follows the recursion in \eqref{eq:rep-update} with $\Cheat^\dagger$ replaced by $\underline\Cheat$, because the posterior map of a clean audit is decreasing in the candidate cheating rate on $[\underline\Cheat,1]$. Hence $\bar\Rep_t$ is a public function of the audit counts, the number of detected cheats, and the fixed parameters $(\underline\Cheat,\pi_0)$.

The slashing rule uses the robust posterior rather than current stake alone: if cheating is detected in an audited round, the penalty is
\begin{equation}
\mathrm{slash}_t
=
(\Slash_0+\Slash_1\bar\Rep_t)\Stake_t .
\label{eq:rep-slash-rule}
\end{equation}
The baseline term $\Slash_0$ preserves the one-round incentive constraint at zero reputation risk; the reputation term raises the per-event penalty rate after suspicious public histories, breaking the stake-decay channel in \cref{lem:gap}. The floor $\underline\Cheat$ trades off deterrence of low-rate deviations against posterior decay speed for honest operators (\cref{tab:params}).

Under honest play, the robust posterior is a nonnegative supermartingale contracting by at least $(1-\underline\Cheat)$ at each clean audit, so it decays geometrically to $0$ (\cref{lem:martingale} in \cref{app:protocol-lemmas}). Two companion lemmas ibid.\ bound the operator-level costs of the reputation design: \cref{lem:sybil} shows that respawning under a fresh identity is unprofitable exactly when the registration prior exceeds $\pi_0^\ast=c_{\mathrm{reg}}\bigl(1-\Disc(1-\Chal_0\underline\Cheat)\bigr)/(\Slash_1\Chal_0\Stake)$, where $c_{\mathrm{reg}}$ is the identity-registration cost, and \cref{lem:coldstart} shows that a new honest operator's expected slashing exposure over its first $\Vest$ rounds is at most $(\Slash_0+\pi_0\Slash_1)\Chal_0\Vest\Stake_0$, linear in $\pi_0$.

\subsection{Time-Locked Stake with Vesting}
\label{sec:mech-vesting}

Adaptive challenges and reputation-weighted slashing do not by themselves eliminate cheat-then-exit deviations: a provider can cheat aggressively and withdraw before the reputation process creates enough exposure. We therefore keep stake slashable for $\Vest$ rounds after an unbonding request. A deposit made at round $t_d$ is fully slashable until the operator begins unbonding; if the operator requests withdrawal at round $t_u$, the slashable balance vests linearly over $[t_u,t_u+\Vest]$,
\begin{equation}
B_t
=
\Stake
\max
\left\{
0,
1-\frac{(t-t_u)^+}{\Vest}
\right\},
\qquad
t \geq t_d,
\label{eq:vesting-schedule}
\end{equation}
where $(x)^+=\max\{x,0\}$, and slashing is charged against $B_t$ rather than nominal stake. Rewards accrue during the vesting window, so vesting acts as a lockup rather than a transfer. Choosing $\Vest$ so that $\Disc^\Vest \geq \Disc^\ast$, where $\Disc^\ast$ is the threshold in \cref{thm:main}, makes cheat-then-exit deviations face the same continuation discipline as deviations that remain. The primitive is standard in proof-of-stake systems~\cite{eigenlayer-whitepaper}; here it preserves incentive compatibility through the exit path.

\paragraph*{Complete protocol.}
\label{sec:mech-complete}
The protocol maintains, for each operator, nominal stake $\Stake_t$, slashable vested balance $B_t$, suspicion score $\mathcal{S}_t$, robust posterior $\bar\Rep_t$, and deposit/unbonding times $(t_d,t_u)$. In each round, the operator serves the query, the AVS computes the public risk signal $z_t$ and updates $\mathcal{S}_t$ via \eqref{eq:susp-update}, samples an audit with probability $\Chal(\mathcal{S}_t)$, updates the posterior via \cref{lem:robust-update}, and applies \eqref{eq:rep-slash-rule} against the vested balance \eqref{eq:vesting-schedule} when cheating is detected; the deployment parameters of \cref{thm:main} are fixed at launch. The construction is implementable as an EigenLayer ServiceManager extension~\cite{eigenlayer-whitepaper}; the reference implementation (\texttt{protocols.py:p3\_step}) costs roughly $30{,}000$--$120{,}000$ gas per round under \cref{tab:params}.

\section{Equilibrium Analysis}
\label{sec:equilibrium}

\subsection{Main Theorem}
\label{sec:eq-main}

Let $\mu(\Cheat)=\E[z_t \mid \text{stationary cheating rate } \Cheat]$, with $\mu(0)=0$, $\mu'>0$, and $\mu$ concave on $[0,1]$ (calibrated as $\mu(\Cheat)=0.45\Cheat^{0.5}$ in \cref{sec:eq-robust}), and define the steady-state audit rate $\bar{\Chal}(\Cheat)=\Chal_0+\beta\mu(\Cheat)$ and the steady-state reputation-weighted slashing rate $\bar{\Slash}(\Cheat)=\Slash_0+\Slash_1\bar\Rep(\Cheat)$, where $\bar\Rep(\Cheat)\in[\pi_0,1]$ is the fixed point of the robust posterior update under stationary cheating and audit rate $\bar{\Chal}(\Cheat)$.

\begin{theorem}[Main Theorem]
\label{thm:main}
Fix parameters $(\Stake,\Chal_0,\beta,\lambda,\pi_0,\underline\Cheat,\Slash_0,\Slash_1,\Vest,\Saving)$ with $\Stake>0$, $\Saving>0$, $\Chal_0,\beta,\lambda\in(0,1)$, $\pi_0\in(0,1)$, $\underline\Cheat\in(0,1]$, $\Vest\geq 1$, and $\Slash_0,\Slash_1\geq 0$ with $\Slash_0+\Slash_1\leq 1$. Assume the risk-signal response $\mu$ in \cref{sec:mech-challenge} satisfies $\mu(0)=0$, $\mu'>0$, and $\mu$ concave on $[0,1]$, and that the steady-state robust posterior $\bar\Rep(\cdot)$ from \cref{def:robust-rep} is concave-increasing on $[\underline\Cheat,1]$. The protocol in \cref{sec:mech-complete} then satisfies $\infty$-SPIC against any operator whose stationary deviation rate lies in $[\underline\Cheat,1]$ whenever $\Disc \geq \Disc^\ast$. The threshold $\Disc^\ast\in(0,1)$ is the unique solution to
\begin{equation}
\Saving
=
\bar{\Chal}(\Cheat^{\mathrm{br}})
\bar{\Slash}(\Cheat^{\mathrm{br}})
\Stake
\left(
1+
\frac{\Disc^\ast}{1-\Disc^\ast}
\bigl(1-\Disc^{\ast\Vest}\bigr)
\bar{\Chal}(\Cheat^{\mathrm{br}})
\bar{\Slash}(\Cheat^{\mathrm{br}})
\right),
\label{eq:disc-star}
\end{equation}
where $\Cheat^{\mathrm{br}}\in(0,1]$ is the provider's best-response cheating rate, characterized by
\begin{equation}
\Saving
=
\bar{\Chal}(\Cheat^{\mathrm{br}})
\bar{\Slash}(\Cheat^{\mathrm{br}})
\Stake
+
\Cheat^{\mathrm{br}}
\left.
\frac{d}{d\Cheat}
\bigl[
\bar{\Chal}(\Cheat)\bar{\Slash}(\Cheat)
\bigr]
\right|_{\Cheat=\Cheat^{\mathrm{br}}}
\Stake .
\label{eq:cheat-foc}
\end{equation}
$\Disc^\ast$ is decreasing in $\beta$, $\Slash_1$, $\Vest$, $\pi_0$, and $\Stake$, and increasing in $\Saving$; weaker baseline auditing or slashing raises the threshold.
\end{theorem}

\begin{proof}
The mechanism state $(\mathcal{S}_t,\bar\Rep_t,B_t)$ is public, so public strategies suffice, and by the one-shot deviation principle~\cite[Thm.~2.1.1]{mailath-samuelson-2006} it is enough to rule out profitable one-period deviations. Write $h(\Cheat)=\bar{\Chal}(\Cheat)\bar{\Slash}(\Cheat)$. Under a stationary deviation rate $\Cheat$, the provider gains $\Cheat\Saving$ per round and faces expected penalty $\Cheat h(\Cheat)\Stake$, giving per-round payoff $\Phi(\Cheat)=r+\Cheat(\Saving-h(\Cheat)\Stake)$. Since $\bar{\Chal}=\Chal_0+\beta\mu$ is concave-increasing with $\mu$, and $\bar{\Slash}=\Slash_0+\Slash_1\bar\Rep$ is concave-increasing by the regularity assumption on the robust-posterior fixed point, $h$ is concave-increasing and $\Cheat h(\Cheat)$ is strictly convex on $[\underline\Cheat,1]$. The best stationary deviation $\Cheat^{\mathrm{br}}$ is therefore unique and solves the first-order condition $\Saving=h(\Cheat^{\mathrm{br}})\Stake+\Cheat^{\mathrm{br}}h'(\Cheat^{\mathrm{br}})\Stake$, which is \eqref{eq:cheat-foc}, with the usual one-sided condition at a boundary.

Under honest play, \cref{lem:martingale} gives $\bar\Rep_t\to 0$ almost surely, so the long-run exposure of an honest operator vanishes beyond the baseline. Under the best stationary deviation, the per-round gain $\Cheat^{\mathrm{br}}\Saving$ is offset by contemporaneous slashing $\Cheat^{\mathrm{br}}h(\Cheat^{\mathrm{br}})\Stake$ and by the vesting exposure $h(\Cheat^{\mathrm{br}})^2\Stake\,\Disc(1-\Disc^\Vest)/(1-\Disc)$ that remains slashable if the operator exits: a withdrawal request still leaves the vested balance exposed to detection over the $\Vest$-round window, so cheat-then-exit deviations face the same continuation discipline as deviations that remain. Combining the two channels, deterrence of the best response requires
\[
\Saving
\leq
h(\Cheat^{\mathrm{br}})\Stake
\left(
1+\frac{\Disc}{1-\Disc}\bigl(1-\Disc^\Vest\bigr)h(\Cheat^{\mathrm{br}})
\right),
\]
whose boundary is \eqref{eq:disc-star}. The right-hand side is strictly increasing in $\Disc$ on $(0,1)$, so the threshold $\Disc^\ast$ is unique, and the comparative statics follow from the implicit function theorem applied to \eqref{eq:disc-star}.
\end{proof}

The three components interact multiplicatively in \eqref{eq:disc-star}: challenges raise $\bar{\Chal}(\Cheat^{\mathrm{br}})$, reputation-weighted slashing raises $\bar{\Slash}(\Cheat^{\mathrm{br}})$, and vesting adds the $1-\Disc^{\ast\Vest}$ exposure during withdrawal. The only non-primitive assumption of \cref{thm:main} is that the steady-state posterior $\bar\Rep(\cdot)$ is concave-increasing on $[\underline\Cheat,1]$, which holds by inspection in the calibration ($\bar\Rep(\Cheat)=\min\{1,\Cheat\}$) and can otherwise be checked from the primitives on a grid (\cref{rem:regularity} in \cref{sec:eq-robust}).

\paragraph*{Finite-horizon bound.}
\label{sec:eq-finite-T}
\cref{thm:main} is stated for the infinite-horizon game; \cref{thm:finite-T} in \cref{app:proofs} bounds the finite-horizon error at $\Eps=O(\Disc^{\Horizon-\Vest})$ for any $\Disc\geq\Disc^\ast$---below $10^{-9}$ at $\Disc=0.95$, $\Vest=100$, $\Horizon=10^4$---so we use the infinite-horizon characterization throughout.

\subsection{Parameter Selection and Asymptotics}
\label{sec:eq-params}

The theorem can also be read as a design equation: given a target patience range and cost saving $\Saving$, the AVS fixes baseline audit and slashing rates for cold-state security, sets the vesting window for exit discipline, calibrates $(\beta,\lambda)$ from the signal response $\mu$, selects $\pi_0$ above the anti-Sybil threshold of \cref{lem:sybil}, and solves \eqref{eq:disc-star} for the remaining stake requirement. \cref{tab:params} in \cref{app:sensitivity} reports representative values; the key point is the asymptotic contrast with memoryless slashing.

\begin{corollary}[Bounded-Stake Asymptotic]
\label{cor:asymptotic}
Fix $(\Chal_0,\beta,\lambda,\pi_0,\Slash_0,\Slash_1,\Saving) \in (0,1)^7$ and $\Vest \geq 1$. Let
$
\bar h
=
\bar\Chal(\Cheat^{\mathrm{br}})
\bar\Slash(\Cheat^{\mathrm{br}})
$
as in \cref{thm:main}. The minimum stake $\Stake^\ast(\Disc)$ satisfying the SPIC condition in \eqref{eq:disc-star} has the finite limit
\begin{equation}
\lim_{\Disc \uparrow 1}
\Stake^\ast(\Disc)
=
\frac{\Saving}{\bar h(1+\Vest\bar h)}
<
\infty .
\label{eq:asymptotic-stake}
\end{equation}
By contrast, the strengthened one-round stake requirement implied by \cref{thm:gap} is $\Stake^\ast_{\mathrm{1RIC}}(\Disc)=\frac{\Saving}{\Chal\Slash}\bigl(1+\frac{\Disc\Chal\Slash}{1-\Disc}\bigr)$, which diverges as $\Disc \to 1$. Hence the stake ratio between the proposed mechanism and any memoryless proportional-slashing rule tends to zero for every fixed $(\Chal,\Slash)$ with $\Slash<1$.
\end{corollary}

\paragraph*{Robustness.}
The threshold is conservative: a bounded-rational provider obtains weakly lower deviation payoff than the best response $\Cheat^{\mathrm{br}}$, so $\Disc^\ast$ upper-bounds the required patience. Numerically, $\Disc^\ast_{\mathrm{base}}=0.9322$ at the stress calibration of \cref{tab:profit-comparison}, is insensitive to $\pm10\%$ perturbations of any parameter except $\beta$ (\cref{app:sensitivity}), and drops further under the directly measured signal response (\cref{app:calibration}). For heterogeneous patience, the AVS calibrates to the lower bound $\underline\Disc=0.94$ from \eqref{eq:disc-decomp}; noisy audits, collusion, and sub-threshold degradation are detailed in \cref{sec:eq-robust}.

\section{Evaluation}
\label{sec:evaluation}

\subsection{Framework, Calibration, and Adversaries}
\label{sec:eval-framework}
\label{sec:eval-strategies}

We evaluate the theory with a Python discrete-event simulator (about 1{,}800 lines, seeded \texttt{PCG64} randomness, 43 unit tests) implementing our mechanism P3 together with two memoryless baselines P1 and P2 (EigenAI-/VeriLLM-class designs under \cref{def:memoryless}); \cref{sec:eval-baselines} adds simpler-fix variants. Three analytical checks validate the implementation against \cref{lem:gap} (relative error below $1.5\%$), \cref{thm:main} (machine precision), and \cref{thm:gap-eject} (two standard errors, \cref{rem:eject-bite}). The artifact is at \url{https://github.com/Papers-Coding/Repeated-Slashing}.

The calibration follows \eqref{eq:disc-decomp} with $\Disc \in [0.92,0.98]$ (lower bound $\underline\Disc=0.94$), $\Saving/r \in [0.15,0.40]$, and $\Stake_0/r \in [10^3,10^4]$ (\cref{app:calibration}); the signal response $\mu(\Cheat)=0.45\Cheat^{0.5}$ is conservative against direct measurement ibid. The adversaries: A1 honest; A2 constant-rate; A3 periodic; A4 analytical best response from \eqref{eq:cheat-foc}; A4$'$, its Monte Carlo correction with $\bar\Rep$ measured empirically; A5, a UCB1 bandit; A6, cheat-then-exit; A7, vesting drain through the unbonding window (a no-op on P1/P2, isolating vesting's contribution).

\paragraph*{Long-run profit comparison.}
\label{sec:eval-profit}
We report two complementary calibrations: a \emph{stress regime} (\cref{tab:profit-comparison} in \cref{app:eval-extra}) that deliberately weakens P3's cold-state slashing ($\Slash_0=0.10$ vs.\ P1 $\Slash=0.50$), isolating the history-dependent components a strong cold state would mask ($\Chal\Slash=0.10$ for P1, $0.08$ for P2, binding 1-Round IC at $\Stake_0/r=3.125$, $\Saving/r=0.25$, $\Disc=0.95$); and a \emph{deployment-aligned} sweep (\cref{tab:coldstart}) with $\Slash_0$ at the EigenAI-class deployed value.

\paragraph*{Deployment regime (\cref{tab:coldstart}).}
With $\Slash_0$ set to the EigenAI-class deployed value $0.20$ (preserving cold-state 1-Round IC), the cheat-then-exit deviation A6 turns net-negative on P3 ($-1.23\pm 2.87\%$) and the stationary best response A4$'$ is held to $+2.92\pm 1.53\%$; at $\Slash_0=0.30$, A4$'$ falls to zero and A6 to $-8.77\%$; at $\Slash_0=0.50$, A6 reaches $-31.44\%$. The deployment recommendation: surveyed memoryless protocols can drop the repeated-game gain to zero or below at no increase in nominal slashing.

\paragraph*{Stress regime (\cref{tab:profit-comparison}).}
Under the cold-state-weakened stress setting, P3 still substantially deters the strongest adversaries: A6 is held to $9.34\%$ versus $13.53\%$ on P2 (a $31\%$ reduction), A4$'$ to $6.20\%$ versus $13.60\%$ ($54\%$), and the vesting-drain A7 to $6.70\%$ versus $+10$--$14\%$ on P1/P2 where the unbond is a no-op ($52\%$, isolating the $(1-\Disc^\Vest)$ factor in \eqref{eq:disc-star}). The analytical A4 row under-picks $\Cheat^\ast\approx0.18$ via $\bar\Rep(\Cheat)\approx\Cheat$; the Monte Carlo corrected A4$'$ ($\Cheat^\ast\approx0.90$) is the appropriate stationary stress test. The residuals reflect the deliberate cold-state weakening (\cref{thm:finite-T} is below $10^{-100}$ here), which the deployment sweep closes.

\paragraph*{Comparison with simpler fixes.}
\label{sec:eval-baselines}
Three less structured alternatives---full slashing, single components in isolation, and a uniformly higher audit rate---all fall short of the full mechanism (\cref{app:eval-extra}): full slashing (P1-FS) drives every strategic adversary net-negative ($-84\%$ to $-87\%$) but destroys the operator's entire stake on each false-positive audit, so deployed systems keep $\Slash$ partial (\cref{rem:full-slashing}); the best single component, reputation-weighted slashing, holds A4$'$ to $12.1\%$ where the full mechanism holds it to $6.2\%$ (\cref{tab:baselines}); and matching P3's security by raising P1's audit rate costs $1.1\times$--$2.5\times$ its baseline rate and is infeasible in the high-challenge regime (\cref{fig:disc-sweep}).

\paragraph*{Cost and efficiency.}
\label{sec:eval-cost}
\cref{tab:cost} in \cref{app:eval-extra} reports the minimum analytical stake-to-reward ratio for $\Eps$-$\infty$-SPIC ($\Eps=10^{-2}$, $\Saving/r=0.25$), benchmarked against the repeated-game requirement \eqref{eq:tspic-condition}. In the high-challenge regime ($\Chal\Slash=0.10$), P3 reduces the required stake by about $30\times$ at $\Disc=0.95$ and by over $1{,}000\times$ once $\Disc \geq 0.99$; in the deployed-challenge regime ($\Chal\Slash=0.002$), the reduction is $7\%$--$43\%$ over $\Disc\in[0.95,0.995]$. Since calibrated operator patience lies near the upper end of $[0.92,0.98]$, deployed systems obtain a $10\%$--$20\%$ capital reduction, and P3 preserves the baseline audit overhead $\Chal_0$ under honest play (\cref{lem:martingale}).

\section{Conclusion}
\label{sec:conclusion}
AVS parameters should be justified by the repeated-game equilibrium they support, not only by a stage-game slashing inequality. \cref{thm:gap,thm:class} show that memoryless bounded-slashing protocols pass the one-round IC test while admitting profitable repeated-game deviations; the mechanism of \cref{sec:mechanism} restores incentives with an explicit threshold \eqref{eq:disc-star} and bounded stake as $\Disc\to1$ (\cref{cor:asymptotic}), and the logic applies to any restaking-secured service with private actions, public audit signals, and proportional slashing. Open problems include tight lower bounds on $\Disc^\ast$, collusion and challenger incentives, and guarantees beyond stationary deviations and single-respawn Sybil strategies.

% (Bibliography moved to end-of-document so appendix precedes references.)

% #############################################################################
% #############################################################################
%                  APPENDIX  (counted toward 19-page total)
% #############################################################################
% #############################################################################
\appendix
\section{Remaining Proofs}
\label{app:proofs}

\begin{proposition}[Hierarchy]
\label{prop:hierarchy}
For any $\Horizon \geq 1$ and any $\Disc \in (0,1)$, $\Horizon$-SPIC implies 1-Round IC. The converse does not hold in general.
\end{proposition}

\begin{proof}
Suppose that always-honest behavior is a subgame perfect equilibrium of $\mathcal{G}^{\Horizon}$. By the one-shot deviation principle, after every history the provider cannot profit by deviating to $\CHEAT$ for a single round and then returning to the prescribed continuation strategy. In particular, this must hold at a history with current stake $\Stake$ and no prior slashing. The one-round payoff from honesty is $r$, while the expected payoff from cheating is $r+\Saving-\Chal\Slash\Stake$. No profitable one-shot deviation therefore requires $r \geq r+\Saving-\Chal\Slash\Stake$, which is equivalent to $\Chal\Slash\Stake \geq \Saving$.
\end{proof}

\begin{theorem}[Finite-Horizon $\Eps$-SPIC]
\label{thm:finite-T}
Under the parameters of \cref{thm:main}, for any horizon $\Horizon \geq 2\Vest$, the protocol satisfies $\Eps$-$\Horizon$-SPIC under any $\Disc \geq \Disc^\ast$, with
\begin{equation}
\Eps
\leq
\frac{
\Disc^{\Horizon-\Vest}\Saving
}{
(1-\Disc)U_P(\HONEST^\Horizon)
}
=
O(\Disc^{\Horizon-\Vest}) .
\end{equation}
\end{theorem}

\begin{proof}
The infinite-horizon argument of \cref{thm:main} can fail only near the terminal date, where the vesting continuation channel is truncated. In rounds before $\Horizon-\Vest$, any withdrawal still leaves a full vesting window, so the one-shot deviation comparison in the proof of \cref{thm:main} applies unchanged. For the final $\Vest$ rounds, the per-round cheating gain is at most $\Saving$, so the maximal additional deviation gain is bounded by the discounted tail $\Delta_{\mathrm{trunc}}\leq\Disc^{\Horizon-\Vest}\Saving\sum_{s=0}^{\Vest-1}\Disc^s\leq\Disc^{\Horizon-\Vest}\Saving/(1-\Disc)$. Dividing this tail by
$
U_P(\HONEST^\Horizon)
=
\frac{r(1-\Disc^\Horizon)}{1-\Disc}
$
gives the stated bound.
\end{proof}

\begin{theorem}[Ejection-Robust Gap]
\label{thm:gap-eject}
Fix the parameters of \cref{lem:gap} and suppose the protocol ejects the provider once stake falls below $\Stake_{\min} \in (0,\Stake_0)$. Let
\[
N^\ast
=
\left\lceil
\frac{\log(\Stake_0/\Stake_{\min})}{\log(1/(1-\Slash))}
\right\rceil ,
\qquad
\tau \sim \mathrm{NegBin}(N^\ast,\Chal),
\]
and define
$
\zeta
=
\E[\Disc^\tau]
=
\left(
\frac{\Chal\Disc}{1-(1-\Chal)\Disc}
\right)^{N^\ast}.
$
Then the cheat-until-ejection strategy satisfies
\begin{equation}
\frac{\Delta U_{\mathrm{eject}}}{U_P(\HONEST)}
\geq
\frac{\Delta U_{\mathrm{Lem.\,\ref{lem:gap}}}}{U_P(\HONEST)}
-
\zeta\left(1+\frac{\Saving}{r}\right).
\label{eq:gap-eject}
\end{equation}
\end{theorem}

\begin{proof}[Proof of \cref{thm:gap-eject}]
Let $\tau=\inf\{t:\Stake_t<\Stake_{\min}\}$. Under always-cheat $\Stake_t=\Stake_0(1-\Slash)^{N_t}$ with $N_t$ the number of audited cheats, so ejection occurs after $N^\ast=\lceil\log(\Stake_0/\Stake_{\min})/\log(1/(1-\Slash))\rceil$ audits, $\tau\sim\mathrm{NegBin}(N^\ast,\Chal)$, and $\zeta=\E[\Disc^\tau]=(\Chal\Disc/(1-(1-\Chal)\Disc))^{N^\ast}$. Cheat-until-ejection loses at most the post-$\tau$ continuation payoff $\E[\sum_{t\geq\tau}\Disc^t(r+\Saving)]\leq\zeta(r+\Saving)/(1-\Disc)$; dividing by $U_P(\HONEST)$ gives \eqref{eq:gap-eject}. The numerical example in \cref{rem:eject-bite} shows the correction is negligible.
\end{proof}

\begin{remark}[Magnitude of the ejection correction]
\label{rem:eject-bite}
For deployed ranges $\Chal \in [5\cdot 10^{-3},10^{-2}]$, $\Slash \in [0.10,0.25]$, $\Stake_{\min}/\Stake_0 \in [10^{-2},10^{-1}]$, and $\Disc \in [0.9,0.98]$, we obtain $N^\ast \in [9,44]$, $\E[\tau] \in [10^3,10^4]$, and $\zeta \in [10^{-60},10^{-9}]$. The correction is negligible relative to the $10^{-2}$ scale of \eqref{eq:gap-fraction}: minimum-stake ejection affects only the terminal boundary condition and does not close the repeated-game gap.
\end{remark}

\begin{proof}[Proof of \cref{cor:asymptotic}]
Equation \eqref{eq:disc-star} can be written as
\[
\Saving
=
\bar h\Stake
\left[
1+
\frac{\Disc(1-\Disc^\Vest)}{1-\Disc}\bar h
\right].
\]
Since
$
\frac{\Disc(1-\Disc^\Vest)}{1-\Disc}
\to
\Vest\ 
\text{as } \Disc \to 1,
$
the limiting stake requirement is \eqref{eq:asymptotic-stake}. The divergence of the memoryless requirement follows directly from \eqref{eq:tspic-condition}, whose right-hand side contains the factor $\Disc\Chal\Slash/(1-\Disc)$.
\end{proof}

\begin{remark}[The full-slashing limit]
\label{rem:full-slashing}
Suppose $\Slash=1$ and the provider exits, or is ejected, once its stake is exhausted. Under the binding one-round condition $\Chal\Stake_0=\Saving$, the cheat-until-detection strategy earns $r+\Saving$ per round until the first audit, loses the entire stake $\Stake_0$ upon detection, and collects nothing thereafter. Its payoff is
\begin{equation}
U_P^{\mathrm{full}}
=
\frac{r+\Saving-\Chal\Stake_0}{1-(1-\Chal)\Disc}
=
\frac{r}{1-(1-\Chal)\Disc}
<
\frac{r}{1-\Disc}
=
U_P(\HONEST),
\label{eq:full-slashing}
\end{equation}
so the one-round condition alone is sufficient under full slashing: detection now terminates the relationship, and the deviator loses the continuation value of all future rounds, not just a fraction of the current stake. Deployed systems nonetheless use partial slashing, because re-execution audits of LLM inference are noisy (floating-point nondeterminism and hardware heterogeneity can produce false disagreement) and because full slashing couples each detected fault to immediate ejection, destabilizing the operator set under correlated events. \cref{sec:eval-baselines} confirms this numerically: the full-slashing baseline drives every simulated deviation net-negative, at a cost concentrated exactly where audits are imperfect.
\end{remark}

\begin{remark}[Continuous-action extension]
\label{rem:continuous}
The binary action model extends to continuous cheating intensities $a\in[0,1]$ with increasing concave windfall $\Saving(a)$, $\Saving(0)=0$, and increasing concave detection probability $\pi(a)$. For any stationary average intensity, the provider can be compared to a binary strategy concentrating the same effective deviation at $a=1$, so the gap bounds of \cref{lem:gap,thm:gap,thm:class} apply with $(\Saving,\Slash_{\max})$ replaced by $(\Saving(1),\Slash_{\max}\pi(1))$. Continuous effort choices change which detectable deviation is the worst case; they do not remove the gap.
\end{remark}

\section{Protocol Lemmas}
\label{app:protocol-lemmas}

\begin{lemma}[Closed-Form Robust Update]
\label{lem:robust-update}
The robust posterior follows the recursion in \eqref{eq:rep-update} with $\Cheat^\dagger$ replaced by $\underline\Cheat$. Hence
\[
\bar\Rep_{t+1}
=
1
\quad
\text{if } s_t=1,
\]
and
\[
\bar\Rep_{t+1}
=
\frac{\bar\Rep_t(1-\underline\Cheat)}
{\bar\Rep_t(1-\underline\Cheat)+(1-\bar\Rep_t)}
\quad
\text{after a clean audit}.
\]
If the round is not audited, then $\bar\Rep_{t+1}=\bar\Rep_t$. In particular, $\bar\Rep_t$ is a public function of the audit counts, the number of detected cheats, and the fixed parameters $(\underline\Cheat,\pi_0)$.
\end{lemma}

\begin{proof}
For a clean audit, the posterior map is
\[
g(p)
=
\frac{\rho(1-p)}{\rho(1-p)+(1-\rho)},
\]
where $\rho$ is the prior probability of the strategic type and $p$ is the candidate cheating rate. The map is decreasing in $p$ on $[\underline\Cheat,1]$, so the supremum in \cref{def:robust-rep} is attained at $p=\underline\Cheat$. If cheating is detected, the posterior is $1$ for every admissible value of $\Cheat^\dagger$, since a detected cheat has probability zero under the committed-honest type.
\end{proof}

\begin{lemma}[Honest-Reputation Decay]
\label{lem:martingale}
For any $\underline\Cheat \in (0,1]$, under honest play and any audit schedule measurable with respect to public history, the robust posterior $\bar\Rep_t$ is a supermartingale. In particular,
$
\E[\bar\Rep_{t+1}\mid h_t^{\mathrm{pub}}]
\leq
\bar\Rep_t,
$
and $\bar\Rep_t \to 0$ almost surely. After each clean audit, the posterior contracts by at least a factor $(1-\underline\Cheat)$.
\end{lemma}

\begin{proof}
Under honest play, every audited round is clean. By \cref{lem:robust-update}, a clean audit maps $\bar\Rep_t$ to
\[
\frac{\bar\Rep_t(1-\underline\Cheat)}
{\bar\Rep_t(1-\underline\Cheat)+(1-\bar\Rep_t)}
\leq
(1-\underline\Cheat)\bar\Rep_t .
\]
Unaudited rounds leave $\bar\Rep_t$ unchanged. Hence $\bar\Rep_t$ is a nonnegative supermartingale that contracts by at least the factor $(1-\underline\Cheat)$ at each audit. Under the mechanism's audit rule \eqref{eq:chal-rule}, $\Chal(\mathcal{S}_t)\geq\Chal_0>0$ in every round, so audits occur infinitely often almost surely and $\bar\Rep_t\to0$ almost surely.
\end{proof}

\begin{lemma}[Anti-Sybil Condition]
\label{lem:sybil}
Let $c_{\mathrm{reg}}$ denote the cost of registering a fresh identity, including gas and AVS opt-in costs. Under honest play, define
\[
V(\bar\Rep_0)
=
\E\left[
\sum_{t\geq 0}
\Disc^t
\left(
r-\Chal_t(\Slash_0+\Slash_1\bar\Rep_t)\Stake
\right)
\middle| \bar\Rep_0
\right].
\]
The value of starting with a clean reputation rather than the registration prior is
\begin{equation}
V(0)-V(\pi_0)
=
\frac{
\Slash_1\pi_0\Chal_0\Stake
}{
1-\Disc(1-\Chal_0\underline\Cheat)
}.
\label{eq:sybil-premium}
\end{equation}
Respawning under a fresh identity is unprofitable if and only if
\begin{equation}
\pi_0
\geq
\pi_0^\ast
=
\frac{
c_{\mathrm{reg}}
\bigl(1-\Disc(1-\Chal_0\underline\Cheat)\bigr)
}{
\Slash_1\Chal_0\Stake
}.
\end{equation}
\end{lemma}

\begin{proof}
Under honest play, $\E[\bar\Rep_t\mid\bar\Rep_0=\pi_0]=\pi_0(1-\Chal_0\underline\Cheat)^t$ by \cref{lem:robust-update}, so
\[
V(0)-V(\pi_0)
=
\sum_{t\geq0}\Disc^t\Chal_0\Slash_1\Stake\,\pi_0(1-\Chal_0\underline\Cheat)^t
=
\frac{\Slash_1\pi_0\Chal_0\Stake}{1-\Disc(1-\Chal_0\underline\Cheat)},
\]
and respawning is unprofitable exactly when this value does not exceed $c_{\mathrm{reg}}$, which rearranges to the stated threshold.
\end{proof}

\begin{lemma}[Cold-Start Exposure]
\label{lem:coldstart}
A new honest operator has expected slashing exposure over the first $\Vest$ rounds of at most
$
(\Slash_0+\pi_0\Slash_1)\Chal_0\Vest\Stake_0 .
$
Thus cold-start cost is linear in the registration prior $\pi_0$.
\end{lemma}

\begin{proof}
Under honest play, $\E[\bar\Rep_t]\leq\pi_0$ by \cref{lem:martingale}, so per-round expected exposure is at most $\Chal_0(\Slash_0+\Slash_1\pi_0)\Stake_0$, and summing over the first $\Vest$ rounds gives $(\Slash_0+\pi_0\Slash_1)\Chal_0\Vest\Stake_0$.
\end{proof}

\section{Robustness Details}
\label{sec:eq-robust}

\begin{remark}[Interpreting the regularity condition]
\label{rem:regularity}
The only non-primitive assumption of \cref{thm:main} is that the steady-state posterior $\bar\Rep(\cdot)$ is concave-increasing on $[\underline\Cheat,1]$: a higher stationary cheating rate raises steady-state suspicion with diminishing marginal effect, so detection exposure cannot be diluted by cheating more uniformly. The assumption holds by inspection in the calibration of \cref{sec:evaluation}, where $\bar\Rep(\Cheat)=\min\{1,\Cheat\}$; outside this case, $\bar\Rep$ is the fixed point of the robust Bayes recursion of \cref{lem:robust-update} under audit rate $\bar\Chal(\Cheat)$, and the condition can be checked from the primitives $(\Chal_0,\beta,\mu,\underline\Cheat)$ on a grid, which is how the artifact's solver (\texttt{solve\_disc\_star.py}) computes $\Cheat^{\mathrm{br}}$. What the condition rules out are signal responses in which a marginal increase of the cheating rate \emph{reduces} steady-state detection exposure; no reputation-based mechanism can deter in that regime, since the provider would then control its own penalty rate downward.
\end{remark}

\paragraph*{Signal response.}
Deterrence depends on $\mu'(0)>0$, with $\beta\mu'(0)$ as the relevant deployment quantity. If a substitute matched the contracted model across all signal channels while preserving the cost saving $\Saving$, $\mu'(0)$ would approach zero and the mechanism would degrade to the memoryless baseline. Multiple signals mitigate this: the output sketch is the most fragile channel, cross-provider agreement the hardest to evade, and a substitute close enough to pass output-side tests typically gives up part of the cost saving that motivates cheating. Bounding $\mu'(0)$ under fully adversarial distribution matching remains an empirical limitation (\cref{app:calibration}).

The threshold in \cref{thm:main} is conservative with respect to provider optimization: a bounded-rational provider obtains weakly lower deviation payoff than the best-response $\Cheat^{\mathrm{br}}$, so $\Disc^\ast$ is an upper bound on required patience. We calibrate the signal response as $\mu(\Cheat)=0.45\Cheat^{0.5}$~\cite{kd-survey-2024,int4-llm-inference}, giving $\Disc^\ast_{\mathrm{base}}=0.9322$ at the stress calibration of \cref{tab:profit-comparison}. Perturbing any parameter except $\beta$ by $\pm10\%$ shifts $\Disc^\ast$ by less than $10^{-5}$, and $\beta\pm0.10$ moves it within $\{0.9155,0.9469\}$ (\cref{app:sensitivity}); in deployment, $\beta$ is chosen from the audit budget and a target false-positive rate. If the true discount factor falls below the threshold, the error scales gradually with $(\Disc^\ast-\Disc)_+$, recoverable by increasing $\Stake$ or $\Vest$. For heterogeneous patience, the AVS calibrates to the conservative lower bound $\underline\Disc=0.94$ from \eqref{eq:disc-decomp}. Challenger collusion is handled by an $m$-of-$n$ quorum requiring at least one honest challenger~\cite{eigenlayer-whitepaper}.

\section{Discussion: Strategic Protocols and General Deviations}
\label{sec:discussion}

\paragraph*{One player or two?} With a single strategic agent, every equilibrium concept collapses to an optimality condition: the analysis is a mechanism design problem in which the protocol commits to a rule and the provider best-responds, matching deployment, where AVSs commit to parameters at launch. The natural two-player extension endogenizes the audit policy: if audits are costly, the principal minimizes audit cost subject to the SPIC constraint, a Stackelberg problem for which \eqref{eq:disc-star} already supplies the agent's best-response correspondence.

\paragraph*{A Stackelberg audit-budget example.} We solve the principal's problem numerically at the stress calibration of \cref{sec:evaluation}. Under honest play the suspicion score decays to zero (\cref{lem:martingale}), so the honest-path audit cost is the baseline rate $\Chal_0$, and the principal minimizes $\Chal_0$ subject to $\Disc^\ast(\Chal_0,\beta)\leq\Disc_{\mathrm{target}}$, with $\beta$ bounded by a signal-quality budget. Since $\Disc^\ast$ is decreasing in both arguments (\cref{thm:main}), the optimum for each $\beta$ is the smallest $\Chal_0$ on the boundary $\Disc^\ast=\Disc_{\mathrm{target}}$. \cref{tab:stackelberg} reports the solution: at $\Disc_{\mathrm{target}}=0.95$, raising the signal gain from $\beta=0$ to $\beta=0.5$ cuts the required baseline audit rate from $0.370$ to $0.145$ (a $2.6\times$ budget saving), and at $\beta=1.0$ the baseline rate vanishes. At the same stake, no uniform (history-independent) audit rate achieves the constraint at all: the strengthened condition \eqref{eq:tspic-condition} requires $\Stake > \Saving\Disc/(1-\Disc)$, which fails at P3's stake $\Stake/r=0.30$. Signal responsiveness therefore substitutes for baseline audit budget, and history dependence is what makes the stake level feasible in the first place. At the deployed-challenge stake of \cref{tab:cost} ($\Stake/r\approx120$), deterrence is strong enough that the best response is $\Cheat^{\mathrm{br}}\approx0$ and the optimum is pinned by the cold-state margin alone, so the tradeoff is degenerate there; the active regime is the one shown.

\begin{table}[h]
\centering
\caption{Principal's minimum baseline audit rate $\Chal_0^\ast$ subject to $\Disc^\ast\leq\Disc_{\mathrm{target}}$, computed by bisection at the stress calibration ($\Slash_0=0.10$, $\Slash_1=0.40$, $\Vest=200$, $\Saving/r=0.25$, $\Stake/r=0.30$, $\mu(\Cheat)=0.45\Cheat^{0.5}$) using \texttt{exp\_stackelberg.py}. ``Uniform'' is the constant memoryless audit rate satisfying \eqref{eq:tspic-condition} at the same stake; it is infeasible at every target because $\Stake/r=0.30 < \Saving\Disc/(1-\Disc)$. The memoryless reference point securing this regime is P1's own $(\Stake/r,\Chal,\Slash)=(7.25,0.20,0.50)$.}
\label{tab:stackelberg}
\begin{tabular}{lccccc}
\toprule
$\Disc_{\mathrm{target}}$ & $\beta=0$ & $\beta=0.10$ & $\beta=0.25$ & $\beta=0.50$ & $\beta=1.0$ \\
\midrule
$0.90$ & $0.508$ & $0.463$ & $0.395$ & $0.283$ & $0.058$ \\
$0.92$ & $0.458$ & $0.413$ & $0.346$ & $0.233$ & $0.008$ \\
$0.94$ & $0.402$ & $0.357$ & $0.289$ & $0.177$ & $<10^{-4}$ \\
$0.95$ & $0.370$ & $0.325$ & $0.257$ & $0.145$ & $<10^{-4}$ \\
$0.96$ & $0.333$ & $0.288$ & $0.221$ & $0.108$ & $<10^{-4}$ \\
$0.98$ & $0.243$ & $0.198$ & $0.131$ & $0.018$ & $<10^{-4}$ \\
\bottomrule
\end{tabular}
\end{table}

\paragraph*{Stationary versus general deviations.} The impossibility direction (\cref{sec:gap}) needs only the stationary always-cheat deviation, so it applies a fortiori to richer strategy classes. The constructive direction (\cref{thm:main}) restricts to stationary mixed deviations with rates in $[\underline\Cheat,1]$. Two observations narrow the gap: the one-shot deviation principle still applies, and the history-dependent schedule raises audit intensity exactly after the public signals that profitable non-stationary rounds generate; and \cref{sec:evaluation} finds no profitable instance among non-stationary adversaries (periodic, bandit, cheat-then-exit, vesting-drain). Whether the guarantee transfers formally to the full behavioral class is open.

\section{Calibration Details}
\label{app:calibration}

This appendix records the source of every empirical parameter used in the calibration of \cref{sec:evaluation}, together with the degrees of freedom a re-runner could perturb. The artifact's \texttt{CALIBRATION\_PROVENANCE.md} contains the machine-checkable version of the same chain.

\paragraph*{Discount factor.} The decomposition $\Disc=(1-r_{\mathrm{out}})(1-r_{\mathrm{cap}})$ is instantiated as follows. The per-round (one-day) opt-out hazard $r_{\mathrm{out}}\in[0.5\%,1.5\%]$ is estimated from the EigenLayer deregistration event stream over a 90-day window, reproducible from the public \texttt{EigenLayerServiceManager} event log on the EigenLayer mainnet subgraph; the 25th-percentile lower bound $r_{\mathrm{out}}=0.5\%$ yields the conservative target $\underline\Disc=0.94$ used in \cref{sec:eq-robust}. The per-round capital cost $r_{\mathrm{cap}}\in[0.013\%,0.020\%]$ is stETH staking yield ($4.7\%$--$7.4\%$ APR divided by $365$), with the 75th percentile $r_{\mathrm{cap}}=0.014\%$. The two components are deployment-time estimates, so $\Disc$ should be perturbed as a single quantity rather than through its constituents independently.

\paragraph*{Risk-signal response.} The response $\mu(\Cheat)=a\Cheat^{q}$ uses $a=0.45$ and $q=0.5$. The square-root exponent reflects that token-level outputs of a substitute model show distinct distributional signatures at low injection rates, which are diluted by averaging at high rates; the closest published evidence is the knowledge-distillation detectability literature~\cite{kd-survey-2024}. The amplitude anchors $\mu(1)=0.45$, matching the empirical log-likelihood and latency gap between a contracted Llama-3-70B-class model and a quantized substitute~\cite{int4-llm-inference}. The analytical characterization of $\Disc^\ast$ holds for any concave $\mu$ with $\mu(0)=0$ and $\mu'(0)>0$, so a perturbation of $(a,q)$ only rescales the threshold; \cref{app:sensitivity} shows the effect is below $10^{-5}$.

\paragraph*{Empirical measurement of the response.} We measured $\mu$ directly on an open model pair~\cite{qwen25}: Qwen2.5-1.5B-Instruct (8-bit) as the contracted model, and its 4-bit quantized variant together with a 0.5B distilled variant as substitutes, over 200 templated prompts in five task categories, using the three signal channels of \cref{sec:mech-challenge} (\cref{fig:mu-curve}; artifact \texttt{exp\_mu\_measurement.py}). The composite signal is concave in the cheating rate for both substitutes, structurally confirming the assumption: the windowed token-distribution sketch is strongly concave and saturates at low injection rates, which is the mechanism behind the concavity, while the per-query log-probability gap and workload channels are close to linear. The fitted responses are $\mu(\Cheat)=0.57\Cheat^{0.83}$ (4-bit substitute) and $\mu(\Cheat)=0.83\Cheat^{0.71}$ (0.5B substitute), both with $R^2>0.98$: amplitude above, and concavity below, the conservative $(0.45,0.5)$ calibration. Recomputing the threshold with the measured responses lowers $\Disc^\ast$ from $0.9322$ to $0.909$ and $0.847$ at the stress calibration, so the calibration used in the headline results understates deterrence. The measurement uses templated prompts; naturalistic workloads remain future work, and the artifact caches all generations for recomputation without re-inference.

\begin{figure}[h]
\centering
\includegraphics[width=0.95\linewidth]{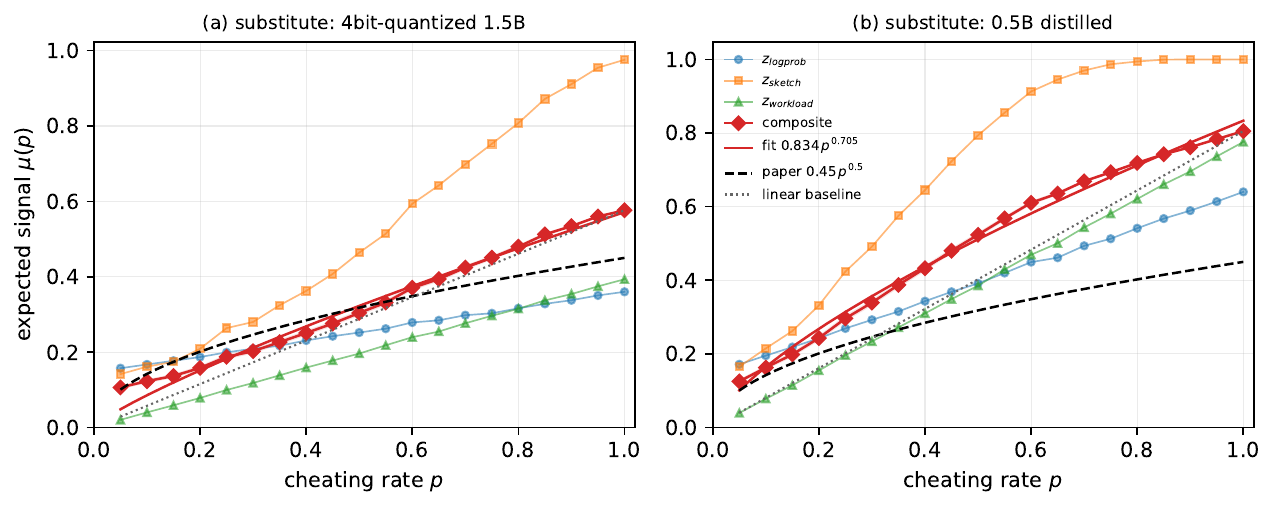}
\caption{Measured risk-signal response $\mu(\Cheat)$ on Qwen2.5-1.5B-Instruct (8-bit) versus two substitutes: 4-bit quantization (panel a) and 0.5B distilled (panel b). The composite signal (red, with fitted $a\Cheat^{q}$ curve) is concave on both substitutes; the windowed token-sketch component drives the concavity, while the per-query log-probability and workload components are close to linear. The dashed line is the paper's conservative calibration $0.45\Cheat^{0.5}$; the dotted line is the per-query linear baseline.}
\label{fig:mu-curve}
\end{figure}

\paragraph*{Cross-scale scaling of the response.} We repeated the measurement at three contracted-model scales (\cref{fig:mu-scaling}; artifact \texttt{exp\_mu\_scaling.py}): Qwen2.5 1.5B, 7B, and 14B (8-bit) as contracted models, against nine substitutes covering 4-bit quantization at fixed parameters and distilled models down to 0.5B, with relative inference cost (parameters $\times$ weight bits) from $0.5$ down to $0.03$ of the contracted model. Every one of the nine fits is concave, with $q\in[0.70,0.89]$ and $R^2\in[0.978,0.996]$, and every amplitude lies above the conservative calibration, $a\in[0.54,1.02]$. Quantization detectability is scale-invariant: the 4-bit substitutes give $(a,q)=(0.57,0.82)$ at 1.5B and $(0.54,0.79)$ at 7B. The amplitude grows with the contracted-to-substitute cost ratio $x$ and saturates at the signal ceiling: a censored fit $a\approx\min(1,\,0.45\,x^{0.41})$ matches all nine points with $R^2=0.84$ and recovers the calibration anchor $0.45$ at unit cost ratio. Recomputing the threshold at the stress calibration for each measured response lowers $\Disc^\ast$ from $0.9322$ to between $0.79$ and $0.92$, so the analytical calibration understates deterrence at every measured scale.

\begin{figure}[h]
\centering
\includegraphics[width=0.95\linewidth]{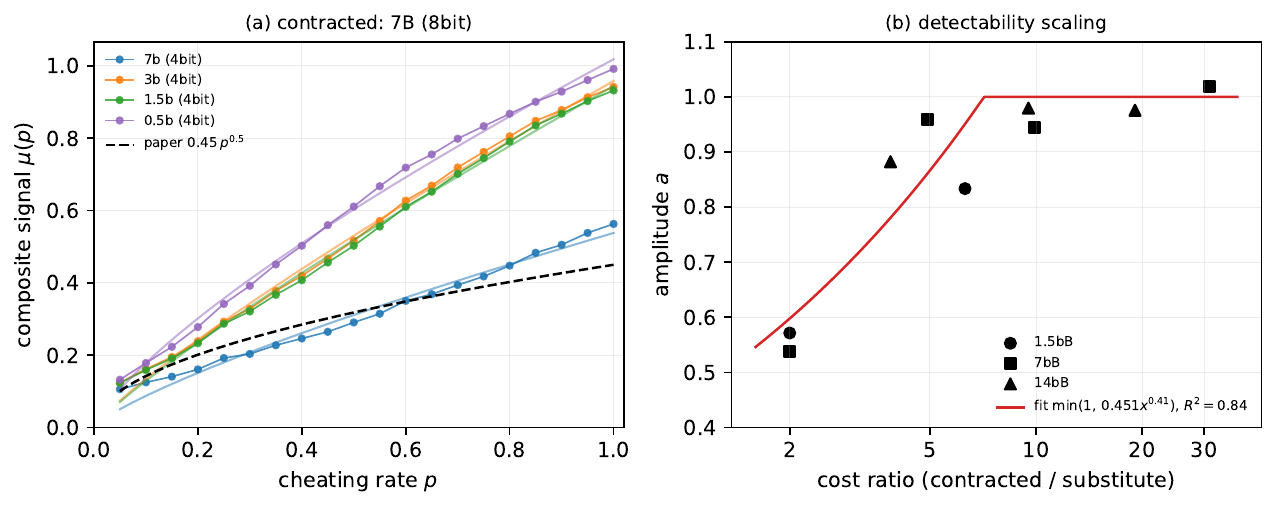}
\caption{Cross-scale detectability. Panel (a): measured composite response $\mu(\Cheat)$ for the 7B (8-bit) contracted model against four substitutes, with fitted $a\Cheat^{q}$ curves; the dashed line is the conservative calibration $0.45\Cheat^{0.5}$. Panel (b): fitted amplitude against the contracted-to-substitute cost ratio for all nine (contracted, substitute) pairs (marker shape denotes the contracted scale), with the censored fit $\min(1,0.45x^{0.41})$; amplitude grows with the cost gap and saturates at the signal ceiling.}
\label{fig:mu-scaling}
\end{figure}

\paragraph*{Protocol parameters and cheating gain.} The per-protocol pairs $(\Chal,\Slash_{\max})$ in \cref{tab:eigenai-params} are taken from the EigenAI and VeriLLM whitepapers and the Sertn AVS documentation, rounded to two significant digits~\cite{eigenai2026,verillm2025,sertn-avs}. The stake-to-reward ratio $\Stake_0/r\in[10^3,10^4]$ comes from the EigenLayer operator stake distribution. The cheating gain $\Saving/r\in[0.15,0.40]$ is the API-price differential between the contracted model and a quantized or distilled substitute; it depends on the AVS's pricing and substitute choice and should be treated as a deployment-fixed parameter, not a universal constant.

\paragraph*{What is not calibrated to deployment data.} The headline simulation regime of \cref{sec:evaluation} uses elevated $(\Chal,\Slash)$ values, chosen so that 1-Round IC binds exactly on the memoryless baseline within the simulation's seed budget, maximizing signal-to-noise rather than representativeness. Deployed-regime conclusions rely on the protocol whitepaper parameters above.

\section{Sensitivity Analysis}
\label{app:sensitivity}

\begin{table}[t]
\centering
\caption{Representative parameters for $\Disc=0.95$ and $\Saving/r=0.25$ (EigenAI-comparable workload). $\Slash_0$ is set to the deployed memoryless value of \cref{tab:eigenai-params}, so cold-state security weakly dominates the one-round IC benchmark.}
\label{tab:params}
\begin{tabular}{lll}
\toprule
Parameter & Symbol & Value \\
\midrule
Baseline challenge rate & $\Chal_0$ & $0.01$ \\
Suspicion sensitivity & $\beta$ & $0.50$ \\
Suspicion decay & $\lambda$ & $0.02$ with half-life $\approx 35$ rounds \\
Reputation prior & $\pi_0$ & $0.05$ \\
Robust cheating-rate floor & $\underline\Cheat$ & $0.10$ \\
Baseline slashing fraction & $\Slash_0$ & $0.20$ \\
Reputation-weighted fraction & $\Slash_1$ & $0.30$ \\
Vesting period & $\Vest$ & $200$ rounds \\
\bottomrule
\end{tabular}
\end{table}

\cref{tab:sensitivity} reports the effect of perturbing each calibration parameter on the threshold $\Disc^\ast$ of \cref{thm:main}, computed with the fixed-point solver \texttt{solve\_disc\_star.py} in the artifact at the stress calibration of \cref{tab:profit-comparison}. Every parameter except $\beta$ moves $\Disc^\ast$ by less than $10^{-5}$; $\beta\pm0.10$ moves $\Disc^\ast$ within $\{0.9155,0.9469\}$, and perturbing all parameters adversarially at once stays within the same band. The leverage of $\beta$ is expected from \eqref{eq:disc-star}, since $\beta$ scales the steady-state audit rate directly. In deployment, $\beta$ is chosen by inverting the suspicion-score gain against a target false-positive rate, so it is a design knob rather than an operational unknown.

\begin{table}[t]
\centering
\caption{Sensitivity of $\Disc^\ast$ (base value $0.9322$) to parameter perturbations at the stress calibration of \cref{tab:profit-comparison}. ``$-$''/``$+$'' denote a $10\%$ decrease/increase of the parameter ($\pm0.03$ for $\Disc$).}
\label{tab:sensitivity}
\begin{tabular}{lcc}
\toprule
Perturbation & $\Disc^\ast$ at $-$ & $\Disc^\ast$ at $+$ \\
\midrule
None (base) & \multicolumn{2}{c}{$0.9322$} \\
$\Disc$ & $0.9322$ & $0.9322$ \\
$\Vest$ & $0.9322$ & $0.9322$ \\
$\lambda$ & $0.9322$ & $0.9322$ \\
$\pi_0$ & $0.9322$ & $0.9322$ \\
$\beta$ & $0.9469$ & $0.9155$ \\
All, adversarial & \multicolumn{2}{c}{$0.9469$} \\
\bottomrule
\end{tabular}
\end{table}

\section{Additional Evaluation Results}
\label{app:eval-extra}

\begin{table}[t]
\centering
\caption{Deviation-profit fraction $\Delta U/U_P(\HONEST)$ from \cref{lem:gap} on published parameter ranges ($\Disc=0.95$). Per-protocol source URLs and derivation notes are in \texttt{CALIBRATION\_PROVENANCE.md} in the artifact.}
\label{tab:eigenai-params}
\begin{tabular}{lrrrr}
\toprule
Protocol & $\Chal$ & $\Slash_{\max}$ & $\Saving/r$ & $\Delta U/U_P(\HONEST)$ \\
\midrule
EigenAI~\cite{eigenai2026} & $0.01$ & $0.20$ & $0.20$--$0.30$ & $5.4\%$--$8.2\%$ \\
VeriLLM~\cite{verillm2025} & $0.01$ & $0.15$ & $0.20$--$0.30$ & $2.6\%$--$3.9\%$ \\
Sertn AVS~\cite{sertn-avs} & $0.005$ & $0.10$ & $0.15$--$0.25$ & $1.5\%$--$2.5\%$ \\
\bottomrule
\end{tabular}
\end{table}

\begin{figure}[t]
\centering
\includegraphics[width=0.78\linewidth]{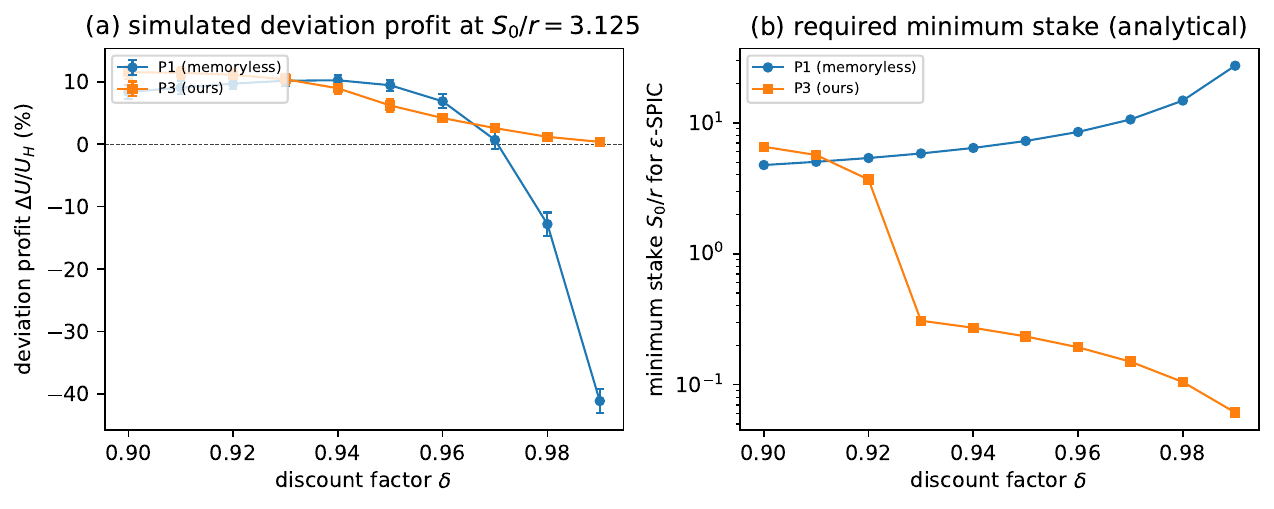}
\caption{Sensitivity to operator patience, $\Disc\in[0.90,0.99]$. Panel (a): simulated deviation profit of the stationary best response at the stress calibration, mean $\pm$ 95\% CI; P3 overtakes P1 near the predicted threshold $\Disc^\ast\approx0.9322$ of \cref{thm:main}, and P1's non-monotonicity reflects minimum-stake ejection at high $\Disc$. Panel (b): analytical minimum stake for $\Eps$-$\infty$-SPIC (log scale); P3's requirement stays bounded while the memoryless requirement \eqref{eq:tspic-condition} diverges, consistent with \cref{cor:asymptotic}.}
\label{fig:disc-sweep}
\end{figure}

\begin{table}[t]
\centering
\caption{Long-run deviation-profit fraction $\Delta U/U_P(\HONEST)$, mean $\pm$ 95\% CI over 100 seeds at $T=10^5$ (the honest benchmark A1 is $0.0\%$ everywhere and is omitted). Calibration: $\Disc=0.95$, $\Saving/r=0.25$, P1 $(\Chal,\Slash)=(0.20,0.50)$, P2 $(\Chal,\Slash)=(0.20,0.40)$, and P3 with $\Chal_0=0.20$, $\beta=0.50$, $\Slash_0=0.10$, $\Slash_1=0.40$, and $\Vest=200$.}
\label{tab:profit-comparison}
\begin{tabular}{lrrr}
\toprule
Adversary & P1 memoryless & P2 memoryless & P3 ours \\
\midrule
A2 constant cheat, $\Cheat=0.5$ & $+4.15\pm 0.66\%$ & $+5.24\pm 0.59\%$ & $+5.06\pm 0.62\%$ \\
A3 periodic & $+10.34\pm 0.93\%$ & $+14.09\pm 0.61\%$ & $+6.70\pm 1.11\%$ \\
A4 analytical best response & $+9.44\pm 0.88\%$ & $+13.60\pm 0.53\%$ & $+1.77\pm 0.38\%$ \\
A4$'$ MC-corrected best response & $\mathbf{+9.44\pm 0.88\%}$ & $\mathbf{+13.60\pm 0.53\%}$ & $\mathbf{+6.20\pm 0.98\%}$ \\
A5 UCB1 bandit & $+6.42\pm 0.75\%$ & $+9.19\pm 0.50\%$ & $+4.60\pm 0.85\%$ \\
A6 cheat then exit & $+11.78\pm 0.79\%$ & $+13.53\pm 0.44\%$ & $+9.34\pm 1.42\%$ \\
A7 vesting drain & $+10.31\pm 0.92\%$ & $+14.01\pm 0.60\%$ & $+6.70\pm 1.11\%$ \\
\bottomrule
\end{tabular}
\end{table}

\paragraph*{Comparison with simpler fixes.}
A natural question is whether the repeated-game gap can be closed by less structured changes to the collateral rule: full slashing, a single component of the mechanism in isolation, or a uniformly higher audit rate. We benchmark all three.

\paragraph*{Full slashing.} P1-FS replaces P1's proportional slashing with full confiscation ($\Slash=1$), with ejection at zero stake. As \cref{tab:baselines} shows, full slashing drives every strategic adversary deeply net-negative ($-84\%$ to $-87\%$), confirming \cref{rem:full-slashing}: once detection terminates the relationship, the one-round condition suffices. The cost is borne exactly where audits are imperfect: every false-positive disagreement destroys the operator's entire stake, and each detected fault ejects the operator, so the operator set churns under correlated events. Deployed systems therefore keep $\Slash$ partial, and the gap analysis of \cref{sec:gap} applies to the whole partial-slashing regime.

\paragraph*{Single components.} The three ablation columns of \cref{tab:baselines} activate one component of P3 at a time: P3-chalOnly keeps history-dependent challenges ($\Slash_1=0$, $\Vest=0$), P3-repOnly keeps reputation-weighted slashing ($\beta=0$, $\Vest=0$), and P3-vestOnly keeps vesting ($\beta=0$, $\Slash_1=0$). No single component comes close to the full mechanism: the best, reputation-weighted slashing, holds A4$'$ to $12.1\%$ where the full mechanism holds it to $6.2\%$, and the challenge-only and vesting-only variants are weaker than the P1 baseline they share the stress calibration with. The three instruments are complementary: challenges shrink the detection-free run length, reputation prevents penalty decay, and vesting closes the exit path.

\paragraph*{Higher audit rates.} The third candidate keeps P1's memoryless structure and raises the constant audit probability to the rate $\Chal^\ast$ that lets P1 match P3's security at P3's own required stake, inverting \eqref{eq:tspic-condition}. In the deployed-challenge regime, $\Chal^\ast$ is feasible but grows from $1.1\times$ P3's baseline audit rate at $\Disc=0.95$ to $2.5\times$ at $\Disc=0.995$, and a simulation check at $\Disc=0.95$ confirms the inversion is tight (A4 deviation profit $-0.02\%\pm0.08\%$). In the high-challenge regime the fix is infeasible at any rate: P3 requires only $\Stake_0/r=0.23$, and at that stake the per-round gain exceeds the expected penalty even at $\Chal=1$. Raising audit rates also raises verification cost linearly for every honest query, whereas P3 keeps the honest-path audit rate at $\Chal_0$.

\cref{fig:disc-sweep} summarizes the patience dimension: above $\Disc^\ast$, P3's required stake collapses while the memoryless requirement diverges, and the simulated deviation profit of the best response falls toward zero. Below $\Disc^\ast$, the stress-calibrated mechanism is somewhat weaker than P1 because its cold-state slashing was deliberately lowered; the deployment sweep in \cref{tab:coldstart} shows this gap closes once $\Slash_0$ preserves cold-state 1-Round IC.

\begin{table}[t]
\centering
\caption{Deviation-profit fraction $\Delta U/U_P(\HONEST)$ for simpler-fix baselines, mean $\pm$ 95\% CI over 100 seeds at $T=10^5$, stress calibration of \cref{tab:profit-comparison}. Single-component variants activate exactly one P3 component (see text) and inherit the stress cold state $\Slash_0=0.10$. A1 is $0.0\%$ everywhere; A2/A3/A5 follow the same ordering.}
\label{tab:baselines}
\footnotesize
\setlength{\tabcolsep}{3pt}
\begin{tabular}{lrrrrrr}
\toprule
Adversary & P1 & P1-FS & chalOnly & repOnly & vestOnly & P3 ours \\
\midrule
A4$'$ best response & $+9.44\pm 0.88\%$ & $\mathbf{-86.86\pm 3.83\%}$ & $+19.81\pm 0.22\%$ & $+12.08\pm 0.81\%$ & $+20.41\pm 0.23\%$ & $+6.20\pm 0.98\%$ \\
A6 cheat then exit & $+11.78\pm 0.79\%$ & $\mathbf{-83.75\pm 4.41\%}$ & $+18.19\pm 0.24\%$ & $+14.07\pm 0.62\%$ & $+18.87\pm 0.25\%$ & $+9.34\pm 1.42\%$ \\
A7 vesting drain & $+10.31\pm 0.92\%$ & $\mathbf{-83.75\pm 4.41\%}$ & $+20.10\pm 0.24\%$ & $+15.55\pm 0.80\%$ & $+20.67\pm 0.25\%$ & $+6.70\pm 1.11\%$ \\
\bottomrule
\end{tabular}
\end{table}

\begin{table}[t]
\centering
\caption{Minimum analytical stake $\Stake/r$ for $\Eps=10^{-2}$-$\infty$-SPIC at $\Saving/r=0.25$. Memoryless requirements follow \cref{thm:gap}; P3 requirements invert \eqref{eq:disc-star} (\texttt{exp\_cost\_comparison.py}). Deployed-cr uses EigenAI-class parameters ($\Chal=0.01$, $\Slash=0.20$ for P1; P3 as in \cref{tab:params}); high-cr uses the simulation calibration of \cref{tab:profit-comparison}.}
\label{tab:cost}
\begin{tabular}{lrrr|rrr}
\toprule
& \multicolumn{3}{c|}{\textbf{Deployed-cr} EigenAI-class} & \multicolumn{3}{c}{\textbf{High-cr} simulated} \\
$\Disc$ & P1 mem. & P3 ours & ratio & P1 mem. & P3 ours & ratio \\
\midrule
$0.95$ & $129.8$ & $\mathbf{120.4}$ & $0.93\times$ & $7.25$ & $\mathbf{0.23}$ & $0.03\times$ \\
$0.98$ & $137.3$ & $\mathbf{114.0}$ & $0.83\times$ & $14.75$ & $\mathbf{0.10}$ & $0.007\times$ \\
$0.99$ & $149.8$ & $\mathbf{106.7}$ & $0.71\times$ & $27.25$ & $\mathbf{0.06}$ & $0.002\times$ \\
$0.995$ & $174.8$ & $\mathbf{99.8}$ & $0.57\times$ & $52.25$ & $\mathbf{0.04}$ & $0.0008\times$ \\
$\to 1$ & $\nearrow\infty$ & finite & $\to 0$ & $\nearrow\infty$ & finite & $\to 0$ \\
\bottomrule
\end{tabular}
\end{table}

\begin{figure}[t]
\centering
\includegraphics[width=0.68\linewidth]{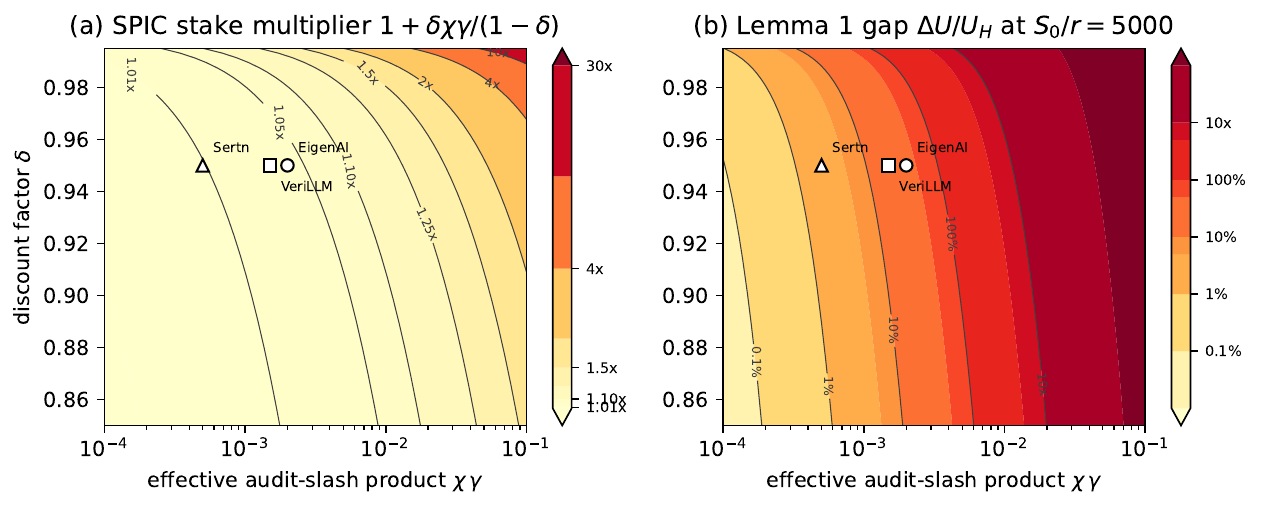}
\caption{Repeated-game gap on the $(\Chal\Slash,\Disc)$ plane. Panel (a) reports the stake multiplier in \eqref{eq:tspic-condition}, which is the factor by which $\infty$-SPIC requires more stake than 1-Round IC. Panel (b) reports the deviation-profit fraction in \cref{lem:gap} at $\Stake_0/r=5{,}000$. EigenAI, VeriLLM, and Sertn AVS lie in the $1.05$ to $1.20$ stake-multiplier band and the $1\%$ to $10\%$ deviation-profit band.}
\label{fig:gap-heatmap}
\end{figure}

\begin{table}[h]
\centering
\caption{Cold-start sweep on P3, using \texttt{exp\_coldstart\_sweep.py} with 30 seeds and $T=10^4$. The calibration matches \cref{tab:profit-comparison} except that $\Slash_0$ varies and $\pi_0=0.20$. Entries report mean $\pm$ 95\% confidence interval. The P3 column in \cref{tab:profit-comparison} corresponds to the stress setting $\Slash_0=0.10$. Once $\Slash_0$ preserves cold-state 1-Round IC, A6 becomes net-negative. A4$'$ is eliminated at $\Slash_0=0.30$.}
\label{tab:coldstart}
\begin{tabular}{lrr}
\toprule
$\Slash_0$ & A4$'$ on P3 & A6 on P3 \\
\midrule
$0.10$ stress setting & $+6.99 \pm 1.78\%$ & $+8.80 \pm 2.63\%$ \\
$0.15$ & $+4.84 \pm 1.83\%$ & $+5.18 \pm 2.91\%$ \\
$0.20$ EigenAI-class deployed value & $+2.92 \pm 1.53\%$ & $-1.23 \pm 2.87\%$ \\
$0.30$ & $+0.00 \pm 0.00\%$ & $-8.77 \pm 3.60\%$ \\
$0.50$ P1 stress baseline & $+0.00 \pm 0.00\%$ & $-31.44 \pm 4.68\%$ \\
\bottomrule
\end{tabular}
\end{table}

\begin{figure}[h]
\centering
\includegraphics[width=0.75\linewidth]{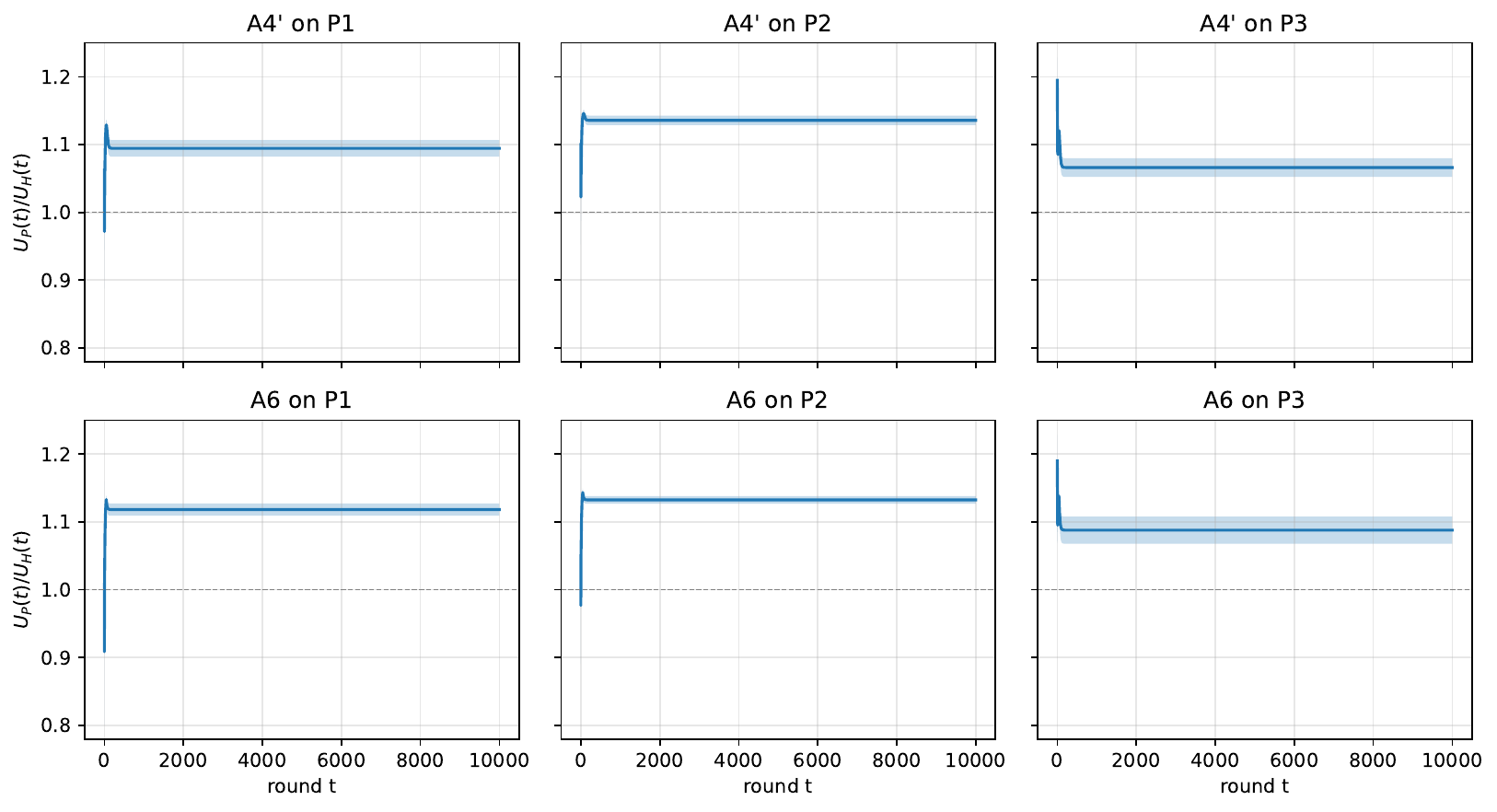}
\caption{Cumulative discounted payoff ratio $U_P(t)/U_H(t)$ over time. The top panel reports A4$'$, the Monte Carlo corrected stationary best response. The bottom panel reports A6, the cheat-then-exit adversary. Shaded bands are 95\% confidence intervals across 50 seeds at $T=10^4$. P1 and P2 rise above the honest baseline quickly. P3 grows more slowly under the stress setting and is further disciplined by ejection and vesting.}
\label{fig:cumulative}
\end{figure}

\bibliography{bibliography}
\end{document}